\documentclass[sigconf]{acmart}
\AtBeginDocument{%
  }

\usepackage{graphicx}
\graphicspath{ {images/} }
\usepackage[font=small,skip=0pt]{caption}
\usepackage[font=small,skip=0pt]{subcaption}

\usepackage{algorithm}
\usepackage[noend]{algpseudocode}
\usepackage{mdwmath,mathtools}
\usepackage{mdwtab}
\usepackage[utf8]{inputenc}
\usepackage[english]{babel}
\usepackage{breqn}
\usepackage{tabto}
\usepackage{stmaryrd}
\usepackage{hyperref}

\usepackage{amsmath,amsfonts}
\usepackage{paralist}
\usepackage[normalem]{ulem}
\usepackage[T1]{fontenc}
\usepackage{booktabs}
\usepackage{multirow}
\usepackage{threeparttable}
\usepackage{multicol}
\usepackage[symbol]{footmisc}
\usepackage{makecell}

\usepackage{color}
\usepackage{url}
\usepackage{bbm}
\usepackage{bm}
\usepackage{graphics}
\usepackage[inline]{enumitem}
\usepackage{tikz-cd}
\usepackage{pgfplots}
\usepgfplotslibrary{groupplots}

\usepackage{tabularx}
\usepackage{xparse}
\usepackage[capitalise]{cleveref}
\usepackage{aliascnt}

\NewDocumentCommand{\statcirc}{ O{#2} m }{%
    \begin{tikzpicture}
    \fill[#2] (0,0) circle (1.0ex); 
    \fill[#1] (0,0) -- (180:1ex) arc (180:0:1ex) -- cycle; 
    \end{tikzpicture}
}

\usepackage{pifont}

\algdef{SE}[EVENT]{Upon}{EndUpon}[1]{\textbf{upon}\ #1\ \algorithmicdo}{\algorithmicend\ \textbf{event}}%
\algtext*{EndUpon}

\algnewcommand{\LineComment}[1]{\State \(\triangleright\) #1}

\usepackage[framemethod=tikz]{mdframed} 
\usepackage{tablefootnote} 
\makeatletter 
\AfterEndEnvironment{mdframed}{%
 \tfn@tablefootnoteprintout%
 \gdef\tfn@fnt{0}%
}
\makeatother 

\usepackage{pgfplots}
\usetikzlibrary{patterns}
\pgfplotsset{compat=1.18}

\usepackage{xspace}
\usepackage{comment}

\newcommand{\sk}{{\sf sk}}
\newcommand{\pk}{{\sf pk}}

\newcommand{\hash}{{\sf Hash}}

\newcommand{\cS} {\mathcal{S}}

\mathchardef\mhyphen="2D

\newcommand{\jolteon}{Jolteon*\xspace}
\newcommand{\name}{Zaptos\xspace}
\newcommand{\baseline}{Aptos Blockchain\xspace}
\newcommand{\latency}{blockchain latency\xspace}

\newcommand{\cons}{{\sf Consensus}\xspace}
\newcommand{\exe}{{\sf Execution}\xspace}
\newcommand{\stg}{{\sf Storage}\xspace}
\newcommand{\stgr}{\ensuremath{{\sf read}}\xspace}
\newcommand{\stgw}{\ensuremath{{\sf write}}\xspace}

\newcommand{\txn}{\ensuremath{{\sf txn}}\xspace}
\newcommand{\txnp}{\ensuremath{{\pi}}\xspace}

\newcommand{\blk}{\ensuremath{{\sf B}}\xspace}
\newcommand{\id}{\ensuremath{{\sf id}}\xspace}
\newcommand{\parent}{\ensuremath{{\sf parent}}\xspace}
\newcommand{\height}{\ensuremath{{\sf height}}\xspace}
\newcommand{\bstate}{\ensuremath{{\sf state}}\xspace}
\newcommand{\payload}{\ensuremath{{\sf payload}}\xspace}
\newcommand{\proposer}{\ensuremath{{\sf proposer}}\xspace}
\newcommand{\pos}{\ensuremath{{\sf pos}}\xspace}

\newcommand{\pipeline}{\ensuremath{\mathcal{P}}\xspace}

\newcommand{\odp}{\ensuremath{\sigma_{\sf od}}\xspace}
\newcommand{\stp}{\ensuremath{\sigma_{\sf st}}\xspace}
\newcommand{\stps}{\ensuremath{\Sigma_{\sf st}}\xspace}
\newcommand{\bsig}{\ensuremath{\sigma_{\sf \blk}}\xspace}

\newcommand{\cmtstate}{\ensuremath{{\sf S_{cmt}}}\xspace}
\newcommand{\cmtheight}{\ensuremath{{\sf H_{cmt}}}\xspace}

\newcommand{\cl}[1]{\ensuremath{{\sf C}_{#1}}\xspace}
\newcommand{\fn}[1]{\ensuremath{{\sf F}_{#1}}\xspace}
\newcommand{\vn}[1]{\ensuremath{{\sf V}_{#1}}\xspace}

\newcommand{\ver}{{\sf Verify}\xspace}
\newcommand{\sign}{{\sf Sign}\xspace}
\newcommand{\agg}{{\sf Aggregate}\xspace}
\newcommand{\veragg}{{\sf VerifyAgg}\xspace}

\newcommand{\preexe}[1]{\colorbox{red!30}{#1}}
\newcommand{\preexebox}[1]{\hspace{-16pt}{\colorbox{red!30}{\hspace{-3pt}\vbox{\hsize .98\linewidth #1}}}}
\newcommand{\precert}[1]{\colorbox{yellow!30}{#1}}
\newcommand{\precertbox}[1]{\hspace{-16pt}{\colorbox{yellow!30}{\hspace{-3pt}\vbox{\hsize .98\linewidth #1}}}}
\newcommand{\precmt}[1]{\colorbox{green!30}{#1}}
\newcommand{\precmtbox}[1]{\hspace{-16pt}{\colorbox{green!30}{\hspace{-3pt}\vbox{\hsize .98\linewidth #1}}}}

\newcommand{\tags}[1]{\textsc{#1}\xspace}

\renewcommand\footnotetextcopyrightpermission[1]{} 

\begin{document}

\title{\name: Towards Optimal Blockchain Latency}


\author{Zhuolun Xiang}
\email{daniel@aptoslabs.com}
\affiliation{%
  \institution{Aptos Labs}
  \country{Palo Alto, USA}
}

\author{Zekun Li}
\email{zekun@aptoslabs.com}
\affiliation{%
  \institution{Aptos Labs}
  \country{Palo Alto, USA}
}

\author{Balaji Arun}
\email{balaji@aptoslabs.com}
\affiliation{%
  \institution{Aptos Labs}
  \country{Palo Alto, USA}
}

\author{Teng Zhang}
\email{teng@aptoslabs.com}
\affiliation{%
  \institution{Aptos Labs}
  \country{Seattle, USA}
}

\author{Alexander Spiegelman}
\email{sasha@aptoslabs.com}
\affiliation{%
  \institution{Aptos Labs}
  \country{Palo Alto, USA}
}








\renewcommand{\shortauthors}{}


\begin{abstract}
    
End-to-end blockchain latency has become a critical topic of interest in both academia and industry. However, while modern blockchain systems process transactions through multiple stages, most research has primarily focused on optimizing the latency of the Byzantine Fault Tolerance consensus component.

In this work, we identify key sources of latency in blockchain systems and introduce Zaptos, a parallel pipelined architecture designed to minimize end-to-end latency while maintaining the high-throughput of pipelined blockchains.

We implemented Zaptos and evaluated it against the pipelined architecture of the Aptos blockchain in a geo-distributed environment. Our evaluation demonstrates a 25\% latency reduction under low load and over 40\% reduction under high load. Notably, Zaptos achieves a throughput of 20,000 transactions per second with sub-second latency, surpassing previously reported blockchain throughput, with sub-second latency, by an order of magnitude.    

\end{abstract}




\maketitle

\section{Introduction}

Modern blockchains are locked in a competitive race to push the boundaries of performance, focusing on maximizing throughput and minimizing latency to meet the demands of scalability and efficiency in decentralized systems.

The majority of research and innovation in both academia and the Web3 industry centers on enhancing the performance of Byzantine Fault Tolerant (BFT) consensus mechanisms~\cite{arun2024shoal++, babel2023mysticeti, doidge2024moonshot}.

However, BFT consensus is neither the throughput nor the latency bottleneck in modern blockchains.
The end-to-end latency, which users care about, is measured from the moment a transaction is submitted to the point of receiving confirmation that it has been committed (applied to the blockchain state). Besides the consensus latency for agreeing on a block of transactions, this process encompasses several other stages: communication time between clients and validators, block execution time, certification of the final execution state, persisting the results to storage, and communicating the outcome back to the client.

In fact, modern consensus systems typically require around 300-400 milliseconds to order a transaction under low load~\cite{aptos, sui}. However, the end-to-end latency of the fastest blockchains is around 1 second in this case~\cite{latency-benchmark, latency-dashboard} and increases substantially as the load rises.

This paper focuses on reducing end-to-end blockchain latency under both low and high load conditions. We propose \name: a parallel
pipelined architecture that maximizes resource utilization for high throughput while parallelizing pipeline stages to achieve optimal blockchain latency.

Theoretically, our approach reduces end-to-end latency by 1 and 5 message hops under low and high loads, respectively. In practice, we demonstrated sub-second end-to-end blockchain latency with a throughput of 20,000 transactions per second (TPS) on a geo-distributed network of 100 validators. This represents a 40\% improvement (over 0.5 second) compared to the \baseline, which serves as the baseline for our work.

It is worth noting that, while much of the community's focus has been on improving consensus latency under high-load conditions, our work achieves a total end-to-end latency reduction that is comparable to the entire consensus latency.

\subsection{Technical overview}

Our baseline is the Aptos blockchain~\cite{aptos}, which builds upon the high-throughput, pipelined architecture introduced by Diem~\cite{diem}. 
For DDoS protection, clients on the Aptos blockchain submit transactions to fullnodes, which then forward them to validators. Validators disseminate the transactions among themselves and reach consensus on a block of transactions. The block is subsequently executed by the execution engine, after which the validators communicate again to certify the final execution result~\footnote{Theoretically, this step is not strictly necessary. However, in practice, the certification stage is crucial for producing state proofs and preventing safety violations in the event of execution divergence due to software bugs.}.
Next, the new state is persisted in the validators' storage, and the block is propagated back to the fullnodes. The fullnodes re-execute the block~\footnote{Fullnodes can alternatively synchronize the transaction outputs or the final state, but in most cases, re-executing the block proves to be faster.} and verify that the final state is certified by the validators. Finally, the fullnodes send transaction confirmations back to the clients.

The pipelined Blockchain design aims to maximize throughput by utilizing all available resources. 
Note that transactions progress through multiple stages in the system, each mainly requiring distinct resources. For instance, the consensus stage is network-intensive, execution relies primarily on the CPU, and persisting data to storage demands significant IO capacity. 
Therefore, a block of transactions does not need to wait for its predecessor to complete all stages before starting its processing. Once a block advances to the next stage, the subsequent block can begin the current stage. As a result, at any given time, block$_i$ is being ordered by consensus, block$_{i-1}$ is being executed, and block$_{i-2}$ is being persisted to storage.

This architecture achieves high system throughput but does not reduce end-to-end latency, as transactions must sequentially pass through all stages of processing.

\name introduces a parallel pipelined architecture designed to reduce end-to-end latency while preserving the high throughput by optimistically \emph{shadowing} most of the pipelining stages.
In a nutshell, the three main optimizations are the following:
\begin{enumerate}
    \item Optimistically execute the block immediately after receiving it in the consensus stage.
    \item Optimistically persist the final state to storage once execution is done without waiting for the state certification phase to complete.
    \item Piggyback the state certification stage onto the consensus protocol. 
\end{enumerate}

Note that the first two optimizations require handling the unhappy path, where consensus fails to order the block due to network issues or a faulty leader. The third optimization is more subtle, as it must prevent a scenario where a Byzantine validator obtains an execution state certificate before the block is ordered by consensus.
This is because, if the block is eventually forked, a safety violation could occur.

Intuitively, to guarantee safety, validators attach their optimistic block execution certificate signatures to their final consensus messages. We prove that if an adversary manages to gather enough signatures to form an execution certificate, the block can no longer be forked and will eventually be ordered.

With these optimizations, \name effectively shadows the execution, state certification, and storage stages under the consensus latency in the common case. This means that by the time a block is ordered, it has already been executed, the final state has been certified, and the data persisted to storage.

Additionally, in \name, validators immediately forward blocks to fullnodes upon first receiving them. This enables fullnodes to apply the first two optimizations, further reducing end-to-end latency.

Hence, the end-to-end latency of \name in this case is equal to 
\vspace{-0.2cm}
\[ 2\delta_{\sf cf} + 2\delta_{\sf fv} + T_{\sf con}, \]  
where \(\delta_{\sf cf}\) represents the communication latency between clients and fullnodes, \(\delta_{\sf fv}\) is the communication latency between fullnodes and validators, and \(T_{\sf con}\) is the consensus latency.  

Notably, clients typically connect to the nearest fullnode, and fullnodes generally synchronize with the closest validator. As a result, \(2\delta_{\sf cf} + 2\delta_{\sf fv}\) is relatively small compared to \(T_{\sf con}\), leaving little room for further latency optimizations.

All in all, the \name parallel pipelined architecture is optimized for both latency and throughput.  
The throughput benefits from efficient pipelined resource utilization, while achieving optimal latency:  
if the blockchain employs an optimal-latency consensus protocol, then the overall blockchain latency is also optimal!

We extended the Aptos open-source with a production-ready implementation of \name and compared both systems in a "mainnet-like" environment with 100 geographically distributed validators.  
Our evaluation demonstrates a 170ms latency reduction under low load and over 0.5s (40\%) latency reduction under high load.  
To the best of our knowledge, \name is the first end-to-end blockchain system to achieve sub-second latency at 20k TPS, surpassing all publicly available Blockchains by more than an order of magnitude.

\section{Preliminary}
\label{sec:prelim}

\subsection{Model Assumptions}
\label{sec:model}
The blockchain system consists of {\em clients}, and servers including {\em validator nodes} (referred to as {\em validators} for brevity) and {\em fullnodes}.
Both validators and fullnodes are geo-distributed. 
Clients, typically users or applications, submit transactions to the system. 
Validators are servers that operate the core blockchain system, ensuring its security as long as a certain threshold of validators behave correctly. 
Fullnodes, which also operate the blockchain system, do not impact the overall correctness of the system but serve to support the network's performance and accessibility.
Clients are restricted to interacting solely with fullnodes for Distributed Denial of Service (DDoS) protection and system scalability.
Fullnodes communicate with both clients and validators, while validators communicate with fullnodes as well as with each other. 
All communication channels are assumed to be reliable and authenticated. 
The results of this paper can be easily extended to a model where clients directly communicate with the validators, as discussed in~\Cref{sec:zaptos:discussion:extension}. 
The system assumes partial synchrony~\cite{dwork1988consensus}, meaning that after an unknown Global Stabilization Time (GST), message delays are bounded by a known upper bound.
We define a {\em round} as a single network delay incurred when a message travels from a sender to a recipient.

For ease of presentation, we assume the standard Byzantine fault tolerance (BFT) model~\cite{lamport1982byzantine}, allowing up to $f$ of the total $n=3f+1$ validators to be malicious. A validator is \emph{honest} if it is not malicious. 
We define a {\em quorum} as any group of $2f+1$ validators. 
The results presented in this paper can be trivially extended to the proof-of-stake~\cite{ethereumpos} setting, where validators hold non-zero stakes and fullnodes hold zero stake.
Any client or fullnode may be malicious, and details of DDoS attack prevention~\cite{zargar2013survey} are omitted here for brevity. 

For simplicity, we assume in this paper that each client connects to its nearest fullnode and each fullnode connects to the closest validator. Additionally, we assume that clients or fullnodes will switch to the next closest peer if the current peer becomes unresponsive or exhibits malicious behavior.
The specific mechanisms for discovering the nearest peers or switching in response to malicious activity are beyond the scope of this paper.

\subsection{Definitions}
\label{sec:definitions}
\begin{definition}[BFT SMR]\label{def:bftsmr}
    A Byzantine fault tolerant state machine replication (BFT SMR) protocol commits client transactions as a log akin to a single non-faulty server, and provides the following guarantees:
    \begin{itemize}
        \item Safety. Honest servers do not commit different transactions at the same log position.
        \item Liveness. Each transaction from honest client is eventually committed by all honest servers.
    \end{itemize}
\end{definition}
The server in the above definition can be any validator or fullnode.
In addition to these requirements, a {\em validated} BFT SMR protocol must ensure {\em external validity}~\cite{cachin2001secure}, where all committed transactions meet an application-specific predicate. This is achieved by validity checks inside consensus, but we omit these details for brevity, focusing instead on the BFT SMR formulation above.
We also assume honest clients will resubmit transactions that are previously failed to commit.

\begin{table}[t]
    \centering
    \small
    \setlength\tabcolsep{8pt}
    \renewcommand{\arraystretch}{1.3}
    \begin{tabular}{cl}
        \toprule
        \textbf{Symbol} & \textbf{Description} \\
        \midrule
        \vn{i}, \fn{i}, \cl{i} & validator node $i$, fullnode $i$, client $i$, respectively \\
        \txn & transaction \\
        \blk & block as $\blk=(\id,\proposer,\height,\parent,\payload,\bsig,$ \\ 
        & $\odp,\bstate,\stps,\stp)$ \\
        \proposer & proposer of the block \\
        \height & block's height in the blockchain \\
        \parent & id of block's parent block \\
        \payload & list of transactions in the block \\
        \id & block's id as $\id=\hash(\proposer || \height || \parent || \payload)$ \\
        \bsig & signature $\bsig=\sign(\blk.\id, \sk_i)$ where $\blk.\proposer=\vn{i}$ \\
        \odp & block's ordering proof as an aggregated signature on $\blk.\id$ \\
        \bstate & blockchain state after block's execution \\
        \stps & set of signatures on $\blk.\bstate$ \\
        \stp & block's state proof as an aggregated signature on $\blk.\bstate$ \\
        \pipeline & local in-memory buffer for blocks that are in the pipeline \\
        \cmtheight & height of latest committed block \\
        \cmtstate & state of latest committed block \\
        \txnp & transaction's inclusion proof, e.g., Merkle tree proof \\
        \bottomrule
    \end{tabular}
    \caption{Notations.}
    \label{tab:symbols}
\end{table}

We summarize the notations used in the paper in~\Cref{tab:symbols}.

\subsection{Primitives}
\label{sec:prim}


\subsubsection*{Multi-signature}\label{sec:prim:mutisig}
We use standard multi-signature scheme such as BLS~\cite{boneh2001short}. As a setup, each validator $\vn{i}$ is assigned with a pair of secret key and public key $(\sk_i, \pk_i)$, and knows the public keys of all other validators and the aggregated public key $\pk$. We use the following standard protocol interfaces for the multi-signature. For brevity, we omit the security definitions, which can be found in~\cite{boneh2001short}.
\begin{itemize}
    \item $\sign(m, \sk_i)\rightarrow \sigma_i$. The signing protocol takes as input a message $m$, and the secret key $\sk_i$ and outputs a signature $\sigma_i$. 
    \item $\ver(\sigma_i, m,\pk_i)\rightarrow 0 / 1$. The verification protocol takes as input a validator $\vn{i}$'s signature $\sigma_i$, a message $m$, and $\vn{i}$'s public key $\pk_i$. It outputs 1 (accept) or 0 (reject).
    \item $\agg(m,\{(\pk_i, \sigma_i)\})\rightarrow \sigma / \bot$. The combine protocol takes as input a message $m$, and a set of tuples $(\pk_i,\sigma_i)$ of public keys and signatures of parties. It outputs an aggregated $\sigma$ or $\bot$.
    \item $\veragg(\sigma, m,\pk)\rightarrow 0 / 1$. The verification protocol takes as input an aggregated signature $\sigma$, a message $m$, and an aggregated public signature $\pk$. It outputs 1 (accept) or 0 (reject).
\end{itemize}

\subsubsection*{Consensus}\label{sec:prim:consensus}
We build the system on top of a consensus protocol that has the following interfaces:
\begin{itemize}
    \item $\cons.input(\txn)$. The consensus protocol takes as input a transaction \txn submitted by a client. 
    \item $\cons.output()\rightarrow \blk$. The consensus protocol outputs a block \blk, where \blk is defined in~\Cref{tab:symbols}~\footnote{When \cons outputs \blk, $\blk.\odp=\blk.\bstate=\blk.\stp=\bot, \blk.\stps=\{\}$.}.
\end{itemize}

\cons guarantees the following properties:
\begin{itemize}
    \item {\em Safety}. If an honest validator outputs $\cons.output()\rightarrow \blk$, then 
    no honest validator outputs $\cons.output()\rightarrow \blk'$ such that $\blk'.\height=\blk.\height$ and $\blk'\neq \blk$.
    \item {\em Liveness}. 
    \begin{itemize}
        \item If an honest validator inputs $\cons.input(\txn)$, then after GST, all the honest validators eventually outputs $\cons.output()\rightarrow \blk$ where $\txn\in\blk$. 
        \item For every block height $h$, after GST, all the honest validators eventually output $\cons.output()\rightarrow \blk$ where $\blk.\height=h$. 
    \end{itemize}
\end{itemize}

In this paper, we consider consensus protocols that have the following properties.
\begin{itemize}
    \item {\em Block Proposal.}
    The protocol has validators proposing block proposals, explicitly or implicitly~\footnote{For leader-based BFT protocols such as~\cite{castro1999practical, gelashvili2022jolteon}, the leader of a given round can explicitly propose a block proposal extending the blockchain. 
    For DAG-based BFT protocols such as~\cite{spiegelman2022bullshark,spiegelman2023shoal,arun2024shoal++,babel2023mysticeti}, there are also chosen leaders that can implicitly propose a block proposal which consists its proposed DAG node and all causally dependent DAG nodes that are not included in the previous implicit block proposal.}.
    To propose a block \blk, a validator $\vn{i}$ broadcasts a message $(\tags{Proposal}, \blk)$.
    \item {\em Order Proof.}
    When outputting a block \blk, the consensus protocol also outputs 
    an order proof $\sigma_\blk$ which is an aggregated signature on the metadata of \blk (e.g., $\blk.\id$) signed by the validators. 
    \item {\em Order Vote.}
    To produce the order proof, as the final step of consensus protocol, each validator $\vn{i}$ broadcasts a message $(\tags{OrderVote}, \blk.\id, \sigma_i=\sign(\blk.\id, \sk_i))$ to vote for ordering a block \blk. A validator orders \blk upon receiving a quorum ($2f+1$) of $\tags{OrderVote}$ messages that aggregate the order proof $\sigma$ for \blk.
\end{itemize}

Numerous consensus protocols satisfy or can be easily adapted to satisfy the order vote property, such as PBFT~\cite{castro1999practical}, Tendermint~\cite{buchman2016tendermint}, SBFT~\cite{gueta2019sbft}, HotStuff~\cite{yin2019hotstuff}, Streamlet~\cite{chan2020streamlet}, Jolteon~\cite{gelashvili2022jolteon}, Moonshot~\cite{doidge2024moonshot} and many others.
Additionally, another series of DAG-based consensus protocols can also satisfy this property, such as Bullshark~\cite{spiegelman2022bullshark}, Shoal~\cite{spiegelman2023shoal}, Shoal++~\cite{arun2024shoal++}, Cordial miners~\cite{keidar2022cordial} and Mysticeti~\cite{babel2023mysticeti}.  

We define a few consensus metrics that will be used in the paper.
\begin{itemize}
    \item {\em Consensus latency}. The consensus latency equals the {\em consensus dissemination latency} plus the {\em consensus ordering latency}. 
    \begin{itemize}
        \item {\em Consensus dissemination latency}. The time duration between any validator $\vn{i}$ calling $\cons.input(\txn)$ and any validator $\vn{j}$ calling $(\tags{Proposal}, \blk)$ such that $\txn\in\blk$.
        \item {\em Consensus ordering latency}. The time duration between any validator $\vn{i}$ calling $(\tags{Proposal}, \blk)$ and any validator $\vn{j}$ calling $\cons.output()\rightarrow (\blk, \sigma)$.
    \end{itemize}
    \item {\em Block interval}. The time duration between any validator $\vn{i}$ calling $\cons.output()\rightarrow (\blk, \sigma)$ and the same validator $\vn{i}$ calling $\cons.output()\rightarrow (\blk', \sigma)$ where $\blk'.\height=\blk.\height+1$.
\end{itemize}

The consensus dissemination latency comprises the time taken to send the transaction to the leader, along with the block queuing latency for the transaction to be included in a block proposal~\footnote{For a consensus protocol with block interval of $T$, the expected block queuing latency is $T/2$ assuming the transactions arrives uniformly at random time. }. If the consensus protocol disseminates transactions in batched fashion, which this paper does, the dissemination latency also includes the batch queuing latency for the transaction to be included in a payload batch.

Usually in the common case, the block interval equals the time duration between two subsequent block proposals.
Several pipelined consensus protocols~\cite{yin2019hotstuff,chan2020streamlet,gelashvili2022jolteon,doidge2024moonshot,spiegelman2022bullshark,spiegelman2023shoal,arun2024shoal++,babel2023mysticeti} chain the block proposals (by hashes) to reuse protocol phases for consecutive proposals for pipelining. As a result, pipelining effectively reduce the block interval.

For the implementation and evaluation (\Cref{sec:eval}), we use \jolteon consensus protocol~\footnote{\jolteon is an improved version of Jolteon~\cite{gelashvili2022jolteon} that reduces consensus dissemination latency~\cite{optqs} and consensus ordering latency~\cite{optjolteon}.} deployed by Aptos blockchain, additionally with techniques from Moonshot~\cite{doidge2024moonshot} to reduce the block interval to a single round.
The resulting consensus protocol satisfies all the aforementioned properties, achieving a consensus dissemination latency of 1.5 rounds plus the batch queuing latency, a consensus ordering latency of 3 rounds, and a block interval of 1 round, while maintaining high throughput and robustness. 
The consensus ordering latency of \jolteon is optimal~\cite{kuznetsov2021revisiting,abraham2021good}). Similar to PBFT~\cite{castro1999practical}, it consists of 1 round of leader broadcasting block proposal, 1 round of vote for \blk, and another round of \tags{OrderVote} for \blk.

\subsubsection*{Execution}\label{sec:prim:execution}
We build the system on top of an execution protocol that has the following interfaces:
\begin{itemize}
    \item $\exe(\bstate, \blk)\rightarrow \bstate_\blk$. The execution protocol takes as input a blockchain state, a block \blk, and outputs the new blockchain state $\bstate_\blk$ after the execution of \blk on top of \bstate. 
\end{itemize}

The execution protocol should be deterministic, i.e., $\exe(\bstate, \blk)$ always produces the same result. 
The blockchain state can have different representations depending on the concrete implementation. In this paper, we assume $\bstate_\blk$ contains inclusion proofs (such as Merkle proof~\cite{merkle1987digital}) for committed transactions and their versions (i.e., position in the committed log) in $\bstate_\blk$. 
    
For the implementation and evaluation (\Cref{sec:eval}), we use BlockSTM~\cite{gelashvili2023block} for parallel execution, and MoveVM~\cite{move} as the smart contract execution engine. 

\subsubsection*{Storage}\label{sec:prim:storage}
We build the system on top of a storage component that has the following interfaces:
\begin{itemize}
    \item $\stg.\stgw(k, v)$. Persist a pair of key and value into the storage. 
    \item $\stg.\stgr(k)\rightarrow v$. Read the value (can be $\bot$) of a given key from the storage. 
\end{itemize}

For simplicity, we abstract the blockchain’s commit phase as persisting the new blockchain state into storage following execution, as the specific implementation details are orthogonal to the opt-commit optimization discussed in this paper (\Cref{sec:zaptos:precmt}). 
In reality, the commit phase may include (1) persisting, per transaction version of the executed block, the transaction itself alone with its side effects
(including events emitted, updates to the state, execution status and error info, gas usage, etc), 
so that such data can be queried by the transaction version; 
(2) calculating and persisting cryptographic summaries (Merkle Trees~\cite{merkle1987digital}) of such data from the executed block, so that any piece of the blockchain raw data is authenticated by an aggregated signature signed by the validators.

For the implementation and evaluation (\Cref{sec:eval}), we use the RocksDB~\cite{rocksdb} as the key-value store and Jellyfish Merkle Tree~\cite{gao2021jellyfish} for data authentication.

\subsection{Metrics}
\label{sec:metrics}
This paper focuses on reducing the {\em end-to-end latency} of the blockchain system while maintaining high {\em throughput}.
The blockchain end-to-end latency (or simply {\em \latency}) refers to the time from when a client submits a transaction to the blockchain system to when the client can receive confirmation that the transaction has been committed (cannot be reverted). 
The throughput represents the number of transactions that can be committed per unit of time, typically measured in transactions per second (TPS).

\section{Pipelined Architecture of Aptos}
\label{sec:baseline}

\begin{figure*}[h!]
    \centering
    \includegraphics[width=0.66\linewidth]{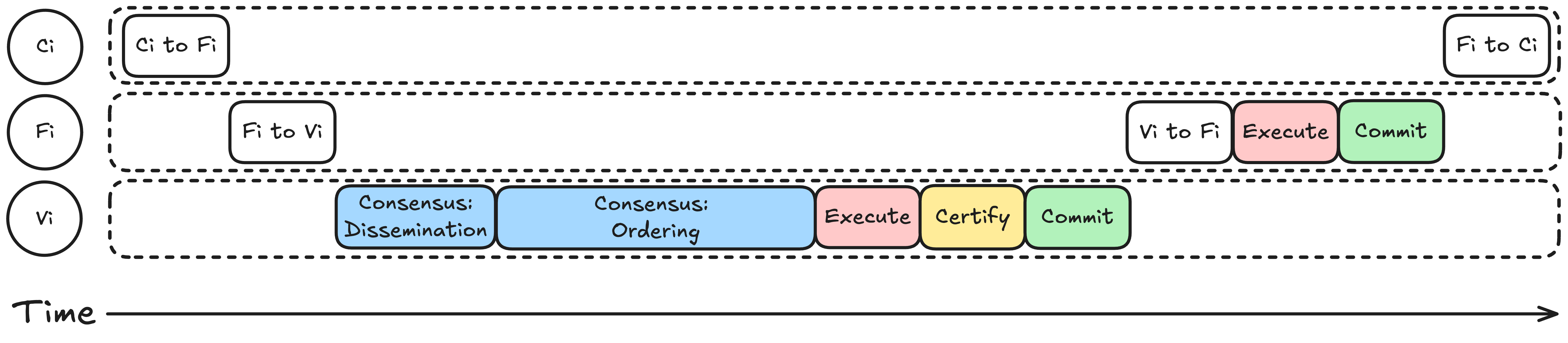}
    \caption{Illustration of the pipelined architecture of modern blockchains such as \baseline~\cite{aptos-whitepaper}. The figure shows client $\cl{i}$, fullnode $\fn{i}$, and validator $\vn{i}$. Each box represents a stage in the blockchain that a block of transactions needs to go through from left to right. The pipeline consists four stages, including consensus (which consists dissemination and ordering), execution, certification and commit. }
    \label{fig:baseline}
\end{figure*}

\begin{figure*}[h!]
    \centering
    \includegraphics[width=0.66\linewidth]{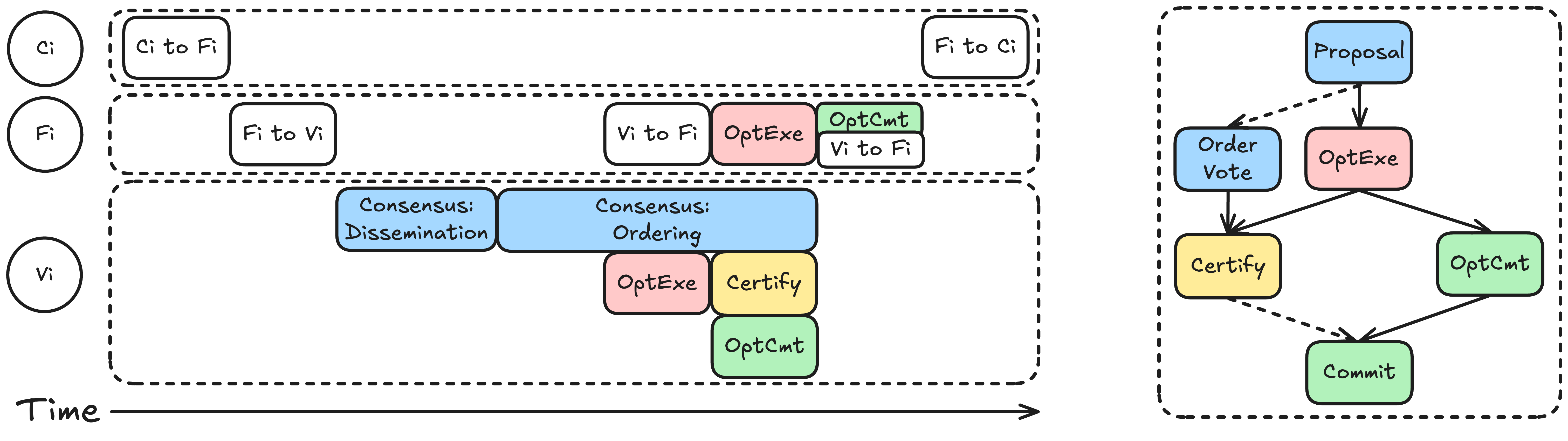}
    \caption{
    Illustration of \name. The left figure illustrates the \name's parallel pipeline architecture. 
    It shows client $\cl{i}$, fullnode $\fn{i}$, and validator $\vn{i}$, and each box represents a stage in the blockchain that a block of transactions needs to go through from left to right. 
    The right figure illustrates the dependencies between stages of the same block.
    Solid arrow denotes direct dependency (a stage starts immediately once all its direct dependencies finishes), dotted arrow denotes indirect dependency (stages connected via more than one direct dependencies). 
    The receipt of block proposal triggers opt-execution. 
    The finishing of opt-execution triggers opt-commit, and triggers certification once \tags{OrderVote} is also sent. 
    When the state is certified and opt-commit finishes, the commit is completed.
    }
    \label{fig:zaptos}
\end{figure*}

\begin{figure}[h!]
    \centering
    \includegraphics[width=\linewidth]{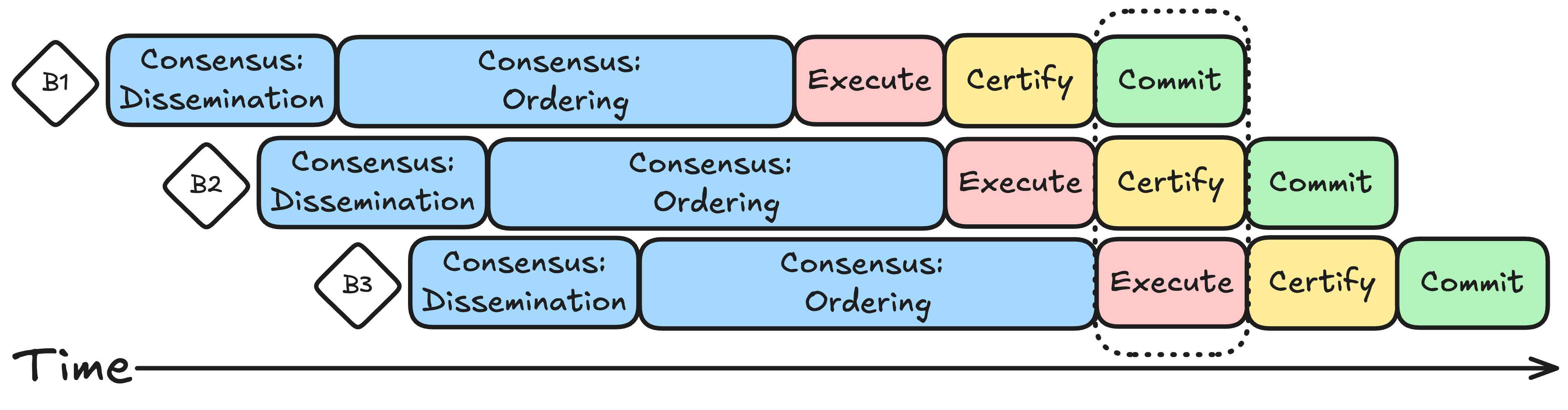}
    \caption{Illustration of the pipelining of consecutive blocks in the \baseline. A validator can pipeline different stages of consecutive blocks, e.g., for blocks $\blk_1,\blk_2,\blk_3$ in the figure, in the dotted slot, the validator can perform in parallel the commit stage  (IO-intensive) of $\blk_1$, the certification stage of $\blk_2$, the execution stage (CPU-intensive) of $\blk_3$, and the consensus stage (network-intensive) of subsequent blocks. 
    In practice, the durations of the stages may vary, leading to imperfect alignment.
    As long as the parallel stages utilize distinct resources, the pipeline improves throughput by maximizing resource utilization compared to non-pipelined designs.
    }
    \label{fig:pipeline-baseline}
\end{figure}

\begin{figure}[h!]
    \centering
    \includegraphics[width=\linewidth]{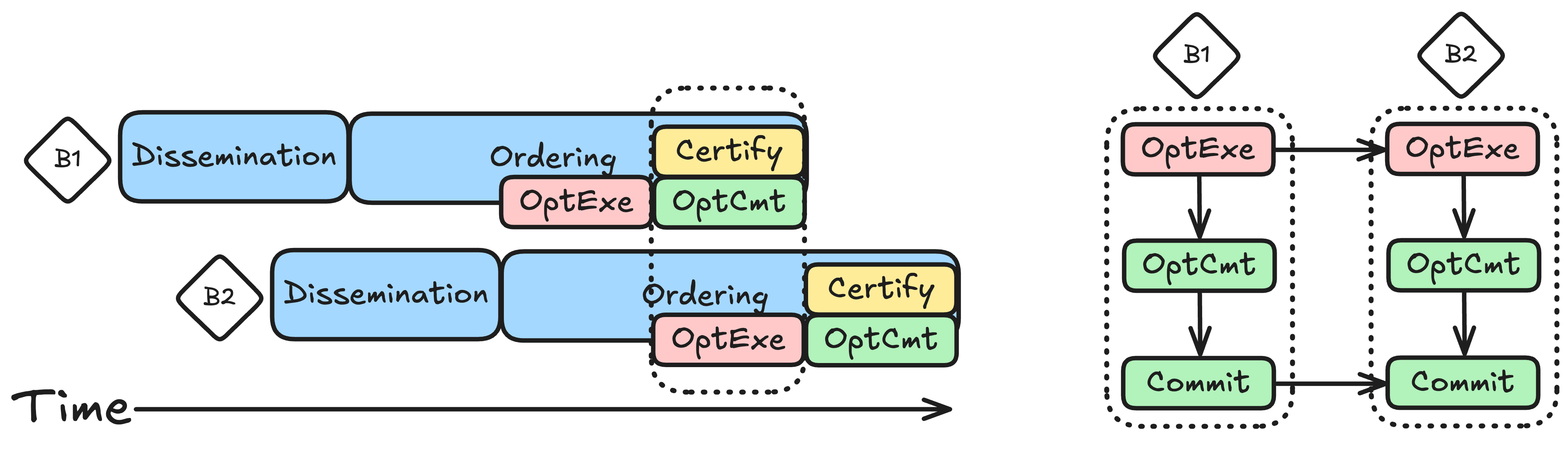}
    \caption{Illustration of the pipelining of consecutive blocks in \name. 
    The left figure illustrates the pipelining.
    Similar to~\Cref{fig:pipeline-baseline}, a validator can also pipeline different stages of consecutive blocks.
    The right figure illustrates the dependencies between stages of consecutive blocks, e.g., the execution and commit stages of block $\blk_2$ also depends on the execution and commit stages of its parent block $\blk_1$, respectively. 
    The boxes representing stages have been slimmed down for presentation purposes only. 
    }
    \label{fig:pipeline-zaptos}
\end{figure}

\begin{algorithm*}[t!]
\caption{Pipelined Architecture of Aptos Blockchain}
\label{alg:baseline}
\begin{multicols}{2}
    \begin{algorithmic}[1]
        \Statex Client $\cl{i}$:
        \Upon{submitting a \txn}
            \State send $(\tags{Submit},\txn)$ to fullnode
        \EndUpon
        \Upon{querying a \txn}
            \State send $(\tags{Query},\txn)$ to fullnode 
            \If{receive $(\tags{Response}, \pos, \txn, \txnp)$ within timeout such that $\txnp.\ver(\pos, \txn)=1$}
                \State \Return (success, \pos)
            \Else
                \ \Return failure
                \Comment{Client may resubmit}
            \EndIf
        \EndUpon
    \end{algorithmic}
    
    \medskip\hrule\medskip
    
    \begin{algorithmic}[1]
    \Statex Fullnode $\fn{i}$:
        \Upon{receiving $(\tags{Submit},\txn)$ from client}
            \State send $(\tags{Submit},\txn)$ to validator
        \EndUpon
        \Upon{receiving $(\tags{Query},\txn)$ from client}
            \If{\txn is committed in position \pos of \cmtstate}
                \State let $\txnp$ be the inclusion proof of $\txn$ in position \pos
                \State send $(\tags{Response}, \pos, \txn, \txnp)$ to client
            \EndIf
        \EndUpon
        \Upon{receiving $(\tags{Committed}, \blk)$ from validator}
            \If{$\blk.\stp=\bot$}
                \Return
            \EndIf
            \State $\pipeline\gets\pipeline\cup\{\blk\}$
            \Comment{Add the block to the pipeline}
        \EndUpon
        \Upon{$\exists\blk\in\pipeline$ s.t. $\blk.\bstate=\bot, \blk'.\bstate\neq\bot$ and $\blk'.\id = \blk.\parent$}
        \Comment{When parent of \blk has been executed}
            \State let $\blk.\bstate\gets \exe(\blk'.\bstate, \blk)$
            \Comment{Execution stage}
        \EndUpon
        \Upon{$\exists\blk\in\pipeline$ s.t. $\blk.\bstate\neq\bot$, $\blk.\stp\neq\bot$ and $\cmtheight = \blk.\height - 1$}
        \If{$\veragg(\blk.\stp, \hash(\blk.\bstate), \pk)=1$}
        \Comment{Commit stage}
            \State let $\cmtheight\gets \blk.\height$, $\cmtstate\gets \blk.\bstate$
            \State $\stg.\stgw(\cmtheight, (\tags{Committed}, \cmtstate))$
        \EndIf
        \EndUpon
            
    \end{algorithmic}
    
    \medskip\hrule\medskip
    
    \begin{algorithmic}[1]
    \Statex Validator $\vn{i}$:
        \Upon{receiving $(\tags{Submit},\txn)$ from Fullnode}
            \State $\cons.input(\txn)$
            \Comment{Consensus stage}
        \EndUpon
        \Upon{$\cons.output()\rightarrow (\blk,\sigma)$}
            \State let $\blk.\odp\gets\sigma$
            \State $\pipeline\gets \pipeline\cup\{\blk\}$
            \Comment{Add the block to the pipeline}
        \EndUpon
        \Upon{$\exists\blk\in\pipeline$ s.t. $\blk.\bstate=\bot, \blk.\odp\neq\bot, \blk'.\bstate\neq\bot$ and $\blk'.\id = \blk.\parent$}
        \Comment{When parent of \blk has been executed}
            \State let $\blk.\bstate\gets \exe(\blk'.\bstate, \blk)$
            \Comment{Execution stage}
            \State let $\sigma_i\gets \sign(\hash(\blk.\bstate), \sk_i)$
            \State let $\blk.\stps\gets \blk.\stps \cup \{(\pk_i, \sigma_i)\}$
            \State send $(\tags{CertifyVote}, \blk.\id, \sigma_i)$ to all validators
        \EndUpon
        
        \Upon{receiving $(\tags{CertifyVote}, id, \sigma_j)$ from validator $\vn{j}$}
            \If{$\exists\blk\in\pipeline$ s.t. $\blk.\id=id$, $\blk.\bstate\neq\bot$, $\blk.\stp=\bot$ and $\ver(\sigma_j, \hash(\blk.\bstate), \pk_j)=1$}
            \Comment{Certification stage}
                \State $\blk.\stps \gets \blk.\stps \cup\{(\pk_j, \sigma_j)\}$
                \If{$|\blk.\stps|\geq 2f+1$}
                    \State let $\blk.\stp\gets\agg(\blk.\bstate,\blk.\stps)$
                \EndIf
            \EndIf
        \EndUpon
        
        \Upon{$\exists\blk\in\pipeline$ s.t. $\blk.\stp\neq\bot$ and $\cmtheight= \blk.\height - 1$}
            \State let $\cmtheight\gets \blk.\height$, $\cmtstate\gets \blk.\bstate$
            \Comment{Commit stage}
            \State $\stg.\stgw(\cmtheight, (\tags{Committed}, \cmtstate))$
            \State $\blk.\bstate\gets\bot$
            \Comment{To save bandwidth}
            \State send $(\tags{Committed}, \blk)$ to fullnode
        \EndUpon
    \end{algorithmic}
  \end{multicols}
\end{algorithm*}

In this section, we present the \baseline~\cite{aptos-whitepaper} architecture, which is a popular blockchain architecture for high-performance blockchains~\cite{diem-whitepaper,aptos-whitepaper,blackshear2024sui,ava-decouple}.
The \baseline employs a {\em pipelined architecture}, allowing different stages of different blocks to execute in parallel. 
The design improves the blockchain performance by maximizing resource utilization.
To the best of our knowledge, the concept of the pipelined architecture was first proposed and productionized by the Diem blockchain~\cite{diem-decouple} (more details in~\Cref{sec:rw:blockchain}), and this paper provides the first comprehensive description.

The original pipeline design of Diem and Aptos blockchains uses doubly linked list to represent the pipeline~\cite{diem-pipeline}, which is conceptually unintuitive and challenging to implement. We simplify the design and implementation using an event-driven framework, as outlined in the pseudo-code in~\Cref{alg:baseline}, and illustrated in~\Cref{fig:baseline}. 
We collaborated with the Aptos Labs team to integrate the simplified pipeline design into the Aptos blockchain~\cite{aptos-pipeline}.
Notably, the event-driven framework also significantly streamlines both the design and implementation of \name, as elaborated in~\Cref{sec:impl}.

The \baseline architecture serves as the baseline for \name, with pseudo-code structured to allow a straightforward description of \name with minimal modifications. 

For brevity, \Cref{alg:baseline} and its description omit implementation details such as message authentication, message verification (e.g., format correctness), error handling (e.g., crash recovery) and garbage collection.

\subsection{The Architecture}\label{sec:baseline:architecture}
We describe the \baseline architecture by tracing the lifecycle of a transaction \txn. 
A client can submit \txn to the fullnode it connects to.
The fullnode, upon receiving \txn from the client, will forward it to the validator it connects to. 
The validator inputs \txn to \cons upon receipt of the transaction.
If the blockchain system is not overloaded, i.e., transaction submission throughput is below the maximum capacity of the blockchain, \cons will order \txn in a few rounds by outputting a block \blk that includes \txn, and its order proof (an aggregated signature $\sigma$ on $\blk.\id$). 
The ordered block \blk is then added to the pipeline (denoted by \pipeline) to go through the following three stages.
\begin{itemize}
    \item {\em Execution stage}. When there exists an unexecuted block \blk in the pipeline such that its parent block has been executed, the validator executes \blk.
    This stage mostly uses CPU resources.
    \item {\em Certification stage}.
    After execution, the validator signs the cryptographic digest of the execution state (represented by $\hash(\blk.\bstate)$ for brevity) and broadcasts the \tags{CertifyVote}. 
    When receiving a \tags{CertifyVote} with valid signature, and the corresponding block \blk has been locally executed but not yet certified, the validator adds the vote to $\blk.\stps$. 
    When the set $\blk.\stps$ reaches a quorum, the validator aggregates the signatures to obtain the aggregated signature that certifies $\blk.\bstate$. 
    This stage uses little computation or bandwidth resources as the \tags{CertifyVote} is small and cheap to compute, but takes one round to receive the \tags{CertifyVote} from a quorum of validators.
    
    The certification stage serves two primary purposes: ensuring safety in cases where the execution stage exhibits non-determinism, and producing a publicly verifiable proof of the new blockchain state which allows fullnodes and clients to independently verify the correctness of the blockchain state. Additional discussions are provided in~\Cref{sec:rw:blockchain}.
    
    \item {\em Commit stage}. If the new certified block is the next in height to be committed, the validator updates the highest committed height and blockchain state, then saves both to storage. 
    This stage mostly uses IO resources. 
    When the commit finishes, the validator sends the newly committed block to the fullnode without including the new state, to reduce bandwidth consumption. 
\end{itemize}

The fullnode, upon receiving the committed block $(\tags{Committed}, \blk)$ from validators, it ensures the state is certified, and adds the block to its pipeline. The pipeline of a fullnode is similar to that of a validator, but without the certification stage. More specifically, the fullnode also executes the blocks~\footnote{
Fullnodes have two options of acquiring the new blockchain state: executing the block locally or synchronizing the execution output from the validators. In practice, the former approach is preferred as it reduces the bandwidth consumption of validators, given that execution outputs can be large. For fullnodes that opt to sync the execution output, \name can also reduce the blockchain's end-to-end latency by enabling earlier synchronization following the validators' optimistic execution.
}, 
and commits the local obtained states once they are certified by the state proof (aggregated signature) provided by the validator, i.e., $\veragg(\blk.\stp, \hash(\blk.\bstate), \pk)=1$.
The client can query whether a transaction \txn is committed on the blockchain. 
The fullnode, upon receiving client's query on \txn, will respond with the inclusion proof \txnp of \txn if \txn is committed in some position \pos according to the latest blockchain state. 
If the client receives the response within a timeout, it verifies if \txnp is valid for \txn, and returns success or failure, respectively.
The client may resubmit the transaction upon failure or timeout. 

\subsection{The Pipelining}\label{sec:baseline:pipeline}
As illustrated in~\Cref{fig:pipeline-baseline}, the pipelined design achieves high blockchain throughput by fully utilizing the different resources of the validators and fullnodes.
Recall that every block progresses through the consensus, execution, certification and commit stages in sequential.
The \baseline architecture uses a pipelined consensus protocol to further reduce the block interval of the consensus stage and therefore reduce the overall blockchain end-to-end latency.
The vanilla \baseline uses Jolteon~\cite{gelashvili2022jolteon}, which has a consensus ordering latency of 5 rounds and a block interval of 2 rounds. 
For the implementation and evaluation of the \baseline and \name in our paper, as mentioned in~\Cref{sec:prim:consensus}, we use the improved version of Jolteon~\cite{gelashvili2022jolteon,optqs,optjolteon} deployed by Aptos blockchain, together with optimizations from Moonshot~\cite{doidge2024moonshot} to reduce block interval. The consensus protocol has an optimal consensus ordering latency of 3 rounds and a block interval of a single round. 

Therefore, as illustrated in~\Cref{fig:pipeline-baseline}, \cons generates a new ordered block every round to enter the pipeline, ideally offset by one round per consecutive block in the pipeline.
To achieve such pipelining, the execution block size can be calibrated to complete execution within approximately one round, matching the duration of the certification stage. 
This alignment allows different stages of consecutive blocks to execute concurrently.
For example, as illustrated in~\Cref{fig:pipeline-baseline}, the commit stage of $\blk_1$, the certification stage of $\blk_2$, and the execution stage of $\blk_3$ all start roughly round the same time and run in parallel with the consensus stage of later blocks. 
Since different stages use distinct resource types (CPU, network bandwidth, IO), the pipelined design maximizes resource utilization, significantly improving the blockchain throughput. 

\section{\name Design}
\label{sec:zaptos}
This section presents the design of \name, which significantly reduces the \latency of the \baseline's pipelined architecture (\Cref{alg:baseline}) through three key optimizations. 
Meanwhile, \name also maintains property of maximizing resource utilization to achieve high throughput. 
The pseudo-code of \name is presented in~\Cref{alg:zaptos}, building on the \baseline's pseudo-code in~\Cref{alg:baseline}. 
Only the differences are highlighted, which may initially appear simple -- reflecting this paper's goal of achieving elegance in architectural design, where protocol simplicity translates to ease of implementation.
Despite this simplicity, the latency improvement is substantial --
as discussed in~\Cref{sec:zaptos:analysis}, \name can reduce the \latency by 5 rounds compared to the \baseline, and can achieve optimal \latency when the underlying consensus protocol also has optimal latency. 

Same as in~\Cref{sec:baseline}, for brevity, \Cref{alg:zaptos} and its description omit implementation details such as message authentication, message verification (e.g., format correctness) and error handling (e.g., crash recovery). 
We will discuss some of the details and considerations of real-world implementation in~\Cref{sec:impl}.

\begin{algorithm}[t!]
  \caption{\name}
  \label{alg:zaptos}
  \begin{minipage}[t]{\columnwidth}
    Fullnode $\fn{i}$: 
    \begin{algorithmic}[1]
        \preexebox{
        \Upon{receiving $(\tags{Proposed}, \blk)$ from validator}
            \State $\pipeline\gets\pipeline\cup\{\blk\}$
            \Comment{Add to pipeline for opt-execution}
        \EndUpon
        }
        \precmtbox{
        \Upon{$\exists\blk\in\pipeline$ s.t. $\blk.\bstate\neq\bot$ and $\cmtheight < \blk.\height$}
            \State $\stg.\stgw(\cmtheight, (\tags{OptCommitted}, \cmtstate))$
            \Comment{Opt-commit}
        \EndUpon
        }
    \end{algorithmic}

    \medskip\hrule\medskip
    
    Validator $\vn{i}$:
    \begin{algorithmic}[1]

        \preexebox{
        \Upon{receiving $(\tags{Proposal}, \blk)$ in \cons}
            \State $\pipeline\gets\pipeline \cup \{\blk\}$
            \Comment{Add to pipeline for opt-execution}
            \State send $(\tags{Proposed}, \blk)$ to fullnode
        \EndUpon
        }
        \precertbox{
        \Upon{$\blk\in\pipeline$, $\blk.\bstate\neq\bot$, $\blk.\odp=\bot$, and
        broadcasting $(\tags{OrderVote}, \blk.\id, \sigma_i=\sign(\blk.\id, \sk_i))$ in \cons}
            \State let $\sigma_i\gets \sign(\hash(\blk.\bstate), \sk_i)$
            \Comment{Certification}
            \State send $(\tags{CertifyVote}, \blk.\id, \sigma_i)$ to all validators
        \EndUpon
        }
        \precmtbox{
        \Upon{$\exists\blk\in\pipeline$ s.t. $\blk.\bstate\neq\bot$ and $\cmtheight < \blk.\height$}
            \State $\stg.\stgw(\cmtheight, (\tags{OptCommitted}, \cmtstate))$
            \Comment{Opt-commit}
        \EndUpon
        }
    \end{algorithmic}
  \end{minipage}
\end{algorithm}

\subsection{Optimistic Execution}
\label{sec:zaptos:preexe}
The change is highlighted in \preexe{red}.
The first optimization, called optimistic execution (or opt-execution), improves the pipeline latency of both validators and fullnodes, by optimistically running the execution stage.

We first explain the changes to the validators. When any validator receives the block proposal $(\tags{Proposal}, \blk)$ in \cons, the validator adds \blk to the pipeline immediately rather than waiting for \blk to be ordered. The validator also sends the proposal to the fullnodes which subscribes to the validator. 
Then, the validator can speculatively execute \blk once the parent block of \blk has been executed, as described in~\Cref{alg:baseline}. 
Since the consensus protocol with optimal consensus ordering latency of three rounds~\cite{castro1999practical,kuznetsov2021revisiting,abraham2021good} consists one round of proposing and two rounds of voting, the optimistic execution allows the validators to start the execution stage of a block in parallel to the second round of the consensus. 
Since the optimistically executed block may not be ordered eventually, if done naively, the validators can still only start the certification stage after the ordering of the block, as in~\Cref{alg:baseline}. The optimization for the certification stage later will address this issue. 

Now we explain the changes to the fullnodes. The fullnode, upon receiving the proposal $(\tags{Proposed}, \blk)$ forwarded by the validator, also adds \blk to the pipeline immediately to speculatively execute \blk once the parent block of \blk has been executed. Same as in the \baseline, the locally computed state is then used for verifying the state proof received from the validator. 
Since the fullnode only commits \blk and responses to the client once \blk's state is certified by the validators, the speculative execution before ordering does not violate safety. 

\subsection{Optimistic Commit}
\label{sec:zaptos:precmt}
The change is highlighted in \precmt{green}.
The second optimization, called optimistic commit (or opt-commit) reduces the commit stage latency for both validators and fullnodes, by allowing blocks to be optimistically committed to storage as soon as the execution stage completes, before the state is certified. 
More specifically, once the execution stage of a block \blk completes, the new state $\blk.\bstate$ to be persisted is already available, allowing the commit stage to initiate (also called opt-commit), abstracted as persisting the value $(\tags{OptCommitted}, \blk.\bstate)$ in the storage for key $\blk.\height$. 
When the validators certified the state, only a minimal update is needed to complete the commit stage, marking the storage entry from $(\tags{OptCommitted}, \blk.\bstate)$ to $(\tags{Committed}, \blk.\bstate)$.
In the case of opt-committed block that is not eventually ordered by consensus, the opt-committed state will be reverted from the storage for data consistency.

\subsection{Certification Optimization}
\label{sec:zaptos:precert}
The change is highlighted in \precert{yellow}.
The final optimization further improves the pipeline latency of the validators, by allowing the validators to start the certification stage of an executed block earlier, rather than waiting for the block to be ordered. 
More specifically, the validators can broadcast $\tags{CertifyVote}$ for a block \blk when broadcasting $\tags{OrderVote}$ for \blk, if \blk is executed. 
This enables the validators to run the certification stage in parallel with the last round of the consensus, effectively reducing the pipeline latency by one round in the common-case. 
Intuitively, the safety of the optimization is guaranteed by the fact that if $\blk.\bstate$ is certified, then \blk is guaranteed to be ordered. We provide the correctness analysis of the optimizations in~\Cref{sec:zaptos:analysis}.

\subsection{Analysis}
\label{sec:zaptos:analysis}

\subsubsection*{Correctness Analysis}\label{sec:zaptos:analysis:correctness}
We prove \name satisfies {\em Safety} and {\em Liveness} defined in~\Cref{def:bftsmr}.

An immediate corollary from the {\em order vote} property and definitions of \cons in~\Cref{sec:prim:consensus}:

\begin{corollary}\label{cor:consensus}
    For any block \blk, if $f+1$ honest validators have broadcasted $(\tags{OrderVote}, \blk.\id, \sigma)$, then all honest validators eventually order \blk after GST.
\end{corollary}

\begin{proof}
    Consider any consensus protocol that satisfies the properties described in~\Cref{sec:prim:consensus}. Consider an execution where $f+1$ honest validators have broadcasted $(\tags{OrderVote}, \blk.\id, \sigma)$. Let the $f$ malicious validators send their $(\tags{OrderVote}, \blk.\id, \sigma)$ message to only a single honest validator $\vn{i}$. By the {\em order vote} property of the consensus protocol, $\vn{i}$ orders \blk. Let $\vn{i}$ be indefinitely partitioned from rest of all honest validators after receiving the $\tags{OrderVote}$ messages and ordering \blk. 
    If any honest validator outputs $\blk'$ such that $\blk'.\height=\blk.\height$ and $\blk'\neq \blk$, the Safety property of the consensus protocol is violated, contradiction. 
    Then by the Liveness property, all honest validators eventually order \blk after GST.
\end{proof}

\begin{theorem}\label{thm:correctness}
    \name as described in~\Cref{alg:zaptos} implements Byzantine fault tolerant state machine replication (BFT SMR) and achieves Safety and Liveness.
\end{theorem}


\begin{proof}
    {\em Safety}.
    Equivalently, we prove that all honest validators and fullnodes commit the same sequence of blocks by showing that, at any height, only a single block is committed. 
    
    First, we prove safety for validators. Suppose for contradiction that two validators commit different blocks $\blk_i \neq \blk_j$ at height $h$. 
    According to the pre-condition of commit, both $\blk_i.\stp$ and $\blk_j.\stp$ were aggregated.
    This implies two sets of $2f+1$ validators, denoted $\cS_i$ and $\cS_j$, must have sent \tags{CertifyVote} messages for $\blk_i$ and $\blk_j$ respectively. Each set contains at least $f+1$ honest validators.
    
    We now examine the cases for when these votes are cast:
    \begin{enumerate}
        \item Suppose all $f+1$ honest validators in $\cS_i$ sent \tags{CertifyVote} for $\blk_i$ when $\blk_i.\odp= \bot$. 
        According to the protocol, these honest validators must also send \tags{OrderVote} for $\blk_i$, implying that all honest validators eventually order $\blk_i$ by~\Cref{cor:consensus}. 
        \begin{enumerate}
            \item\label{item:1:1} If all $f+1$ honest validators in $\cS_j$ also sent \tags{CertifyVote} for $\blk_j$ when $\blk_j.\odp= \bot$, then by~\Cref{cor:consensus}, all honest validators eventually order $\blk_j$, contradicting the Safety property of the consensus protocol.
            \item If at least one honest validator $\vn{j}\in\cS_j$ sent \tags{CertifyVote} for $\blk_j$ when $\blk_j.\odp\neq \bot$, then $\vn{j}$ would output $\blk_j$ in consensus, again contradicting the Safety property.
        \end{enumerate}
        \item Suppose there exists an honest validator $\vn{i}\in \cS_i$ who sent \tags{CertifyVote} for $\blk_i$ when $\blk_i.\odp\neq \bot$.
        \begin{enumerate}
            \item If all $f+1$ honest validators in $\cS_j$ sent \tags{CertifyVote} for $\blk_j$ when $\blk_j.\odp= \bot$, this is symmetric to Case~\ref{item:1:1} and also leads to a Safety violation.
            \item If at least one honest validator $\vn{j}\in \cS_j$ sent \tags{CertifyVote} for $\blk_j$ when $\blk_j.\odp\neq \bot$, then $\vn{j}$ outputs $\blk_j$ in consensus, again contradicting the Safety property.
        \end{enumerate}
    \end{enumerate}
    
    Therefore, for any height $h$, all honest validators aggregate the same $\blk.\stp$ and commit the same block $\blk$. Since only one $\blk.\stp$ exists for each height, fullnodes also commit the same block.

    {\em Liveness}.
    Recall that an honest client will keep resubmitting failed transactions, and clients or fullnodes will switch peers in response to any malicious activity. 
    Eventually, if a transaction \txn has not been committed before GST, the honest client will submit \txn to an honest fullnode after GST, which will then forward \txn to an honest validator. 
    By the Liveness property of the consensus protocol, all honest validators will eventually output the same ordered block \blk containing \txn. 
    Since execution is deterministic, all honest validators will compute the same $\blk.\bstate$ and sign the same message $\hash(\blk.\bstate)$ in \tags{CertifyVote}. 
    Eventually, all honest validators will receive \tags{CertifyVote} message from each other, allowing them to aggregate $\blk.\stp$.
    By safety, for any height $h$, there exists a unique $\blk.\stp$ where $\blk.\height=h$, ensuring that all honest validators will commit \blk by writing the new blockchain state to storage. 
    All honest fullnodes will eventually receive $\blk.\stp$ and subsequently commit \blk.
\end{proof}

\subsubsection*{Latency Analysis}\label{sec:zaptos:analysis:latency}

For the latency analysis for the \baseline's pipelined architecture (\Cref{alg:baseline}) and \name (\Cref{alg:zaptos}), we assume the following:
\begin{itemize}
    \item The client is honest and submits a correctly formatted transaction. 
    \item The client-to-fullnode, fullnode-to-validator and validator-to-validator latencies remain constant, and are denoted as $\delta_{\sf cf}$, $\delta_{\sf fv}$ and $\delta_{\sf vv}$ respectively. The network latency is symmetric for any single connection.
    \item Validators and fullnodes have same consensus, execution and commit stage latencies for a single block, denoted as $T_{\sf con}$, $T_{\sf exe}$ and $T_{\sf cmt}$ respectively. The consensus stage latency $T_{\sf con}$ is the consensus latency mentioned in~\Cref{sec:prim:consensus}.
    \item Any delay in the execution stage due to waiting for the completion of the parent block's execution stage is assumed to be the same in both the \baseline and \name. Consequently, when analyzing \name's improvement over the \baseline, we disregard any execution-stage delays.
    \item We also disregard additional delays caused by “unhappy path” in both the \baseline and \name, such as the client reconnecting to an honest fullnode and resubmitting the transaction, or instances where the consensus re-proposes the block containing the transaction due to previous proposal not getting ordered.
\end{itemize}

\begin{theorem}\label{thm:latency}
    Let $T_{\sf baseline}$ and $T_{\sf zaptos}$ denote the \latency of the \baseline (\Cref{alg:baseline}) and \name (\Cref{alg:zaptos}), respectively. 
    We have
    \begin{align}
    T_{\sf baseline}=2\delta_{\sf cf}+2\delta_{\sf fv}+\delta_{\sf vv}+T_{\sf con}+2T_{\sf exe}+2T_{\sf cmt}
    \end{align}
    \begin{align}
    &T_{\sf zaptos}=2\delta_{\sf cf}+2\delta_{\sf fv}+T_{\sf con}+\max(T_{\sf exe}+T_{\sf cmt}-2\delta_{\sf vv}, 0) \label{eq:zaptos} \\
    & +\max(T_{\sf exe}-\delta_{\sf vv},0)+\max( T_{\sf cmt}-\delta_{\sf vv},0) \nonumber
    \end{align}

    When $T_{\sf exe}\leq \delta_{\sf vv}$ and $T_{\sf cmt}\leq \delta_{\sf vv}$, we have $T_{\sf zaptos}=2\delta_{\sf cf}+2\delta_{\sf fv}+T_{\sf con}$.
\end{theorem}

Results from~\Cref{sec:eval:common:breakdown} under medium load also partially validates~\Cref{thm:latency}, where $T_{\sf zaptos}\approx T_{\sf con}$ when $\delta_{\sf cf}$ and $\delta_{\sf fv}$ are negligible.
An immediately corollary of~\Cref{thm:latency} is~\Cref{cor:latency}. 

\begin{corollary}\label{cor:latency}    
    When $T_{\sf exe}\geq \delta_{\sf vv}$ and $T_{\sf cmt}\geq \delta_{\sf vv}$, we have $T_{\sf baseline}-T_{\sf zaptos}= 5\delta_{\sf vv}$, namely \name improves the \latency of the \baseline by 5 rounds. 
\end{corollary}

\begin{proof}[Proof of~\Cref{thm:latency}]
In the \baseline's pipelined architecture, as shown in~\Cref{fig:baseline}, the \latency consists:
\begin{enumerate}
    \item One client-to-fullnode network delay (denoted as $\delta_{\sf cf}$) and one fullnode-to-validator network delay (denoted as $\delta_{\sf fv}$) to send the transaction.
    \item Validators' consensus stage latency $T_{\sf con}$.
    \item Validators' execution, certification and commit stages latencies, which is $T_{\sf exe}+\delta_{\sf vv}+T_{\sf cmt}$. 
    \item One validator-to-fullnode network delay (same as $\delta_{\sf fv}$) to send certified blocks.
    \item Fullnodes' execution and commit stages latencies (same as $T_{\sf exe}$ and $T_{\sf cmt}$).
    \item One fullnode-to-client network delay (same as $\delta_{\sf cf}$) to send transaction commit confirmation. 
\end{enumerate}

The \latency of the \baseline can be approximated as
\begin{align}
T_{\sf baseline}=2\delta_{\sf cf}+2\delta_{\sf fv}+\delta_{\sf vv}+T_{\sf con}+2T_{\sf exe}+2T_{\sf cmt}
\end{align}

In \name, as shown in~\Cref{fig:zaptos}, the \latency keeps (1), (2), (4) and (6) same as the \baseline, but reduces the latencies of (3) and (5). 
More specifically, for validators in \name, the execution stage starts upon receiving the block proposal (second last round of the consensus ordering), then the certification stage starts upon finishing execution and sending \tags{OrderVote} (last round of the consensus ordering), the commit stage starts upon finishing execution and completes once certification stage finishes.
Hence the duration of validators for all stages will be $T_{\sf vn}=T_{\sf con}-2\delta_{\sf vv}+\max(\delta_{\sf vv}, T_{\sf exe})+\max(\delta_{\sf vv}, T_{\sf cmt})$.
For fullnodes in \name, similarly, the execution stage starts upon receiving the block proposal from the validator (last round of the consensus ordering), and then the commit stage starts upon finishing execution and completes when receiving the committed blocks from the validator. 
Hence the time for fullnodes to commit a block since the validators committed the block will be 
$T_{\sf fn}=\max(\delta_{\sf fv}+T_{\sf exe}+T_{\sf cmt}- 2\delta_{\sf vv}, \delta_{\sf fv})$.
The \latency of \name can be approximated as
\begin{align}
&T_{\sf zaptos}=\delta_{\sf cf}+\delta_{\sf fv}+ T_{\sf vn}+T_{\sf fn}+\delta_{\sf cf}=  
2\delta_{\sf cf}+2\delta_{\sf fv}+T_{\sf con} \label{eq:zaptos} \\
& +\max(T_{\sf exe}-\delta_{\sf vv},0)+\max( T_{\sf cmt}-\delta_{\sf vv},0)+\max(T_{\sf exe}+T_{\sf cmt}-2\delta_{\sf vv}, 0) \nonumber
\end{align}

\end{proof}

\subsection{Discussion}
\label{sec:zaptos:discussion}

\subsubsection*{Latency Optimality}\label{sec:zaptos:discussion:optimality}
As in~\Cref{thm:latency}, when the execution and commit time is bounded by one round, \name has a \latency of $T_{\sf zaptos}=2\delta_{\sf cf}+2\delta_{\sf fv}+T_{\sf con}$. 
Since the round-trip latency of client-to-fullnode and fullnode-to-validator, i.e., $2\delta_{\sf cf}+2\delta_{\sf fv}$, is inevitable, the \latency achieves optimality when the consensus latency $T_{\sf con}$ is optimal. 
As discussed in~\Cref{sec:rw}, for consensus protocols with high throughput and robustness, the consensus latency is at least 4 rounds, even though the lower bound of consensus ordering latency is 3 rounds~\cite{kuznetsov2021revisiting,abraham2021good}.
It remains an interesting open question that if a consensus protocol with optimal 3-round latency can achieve high throughput and robustness, or there exists a latency lower bound of 4 rounds for such protocol. 

\subsubsection*{Throughput}\label{sec:zaptos:discussion:throughput}
In the {\em common case} where the block proposals are ordered by consensus, \name is capable of achieving the \baseline's throughput. 
As illustrated in~\Cref{fig:pipeline-baseline} and~\Cref{fig:pipeline-zaptos}, \name pipelines different stages of consecutive blocks similar to the \baseline. 
Assuming the stages are aligned in both the \baseline and \name, validators perform the consensus, execution, certification, and commit stages of consecutive blocks in parallel at any given time, maximizing resource utilization and supporting high throughput. Experimental evaluations in~\Cref{sec:eval} will support the theoretical analysis of throughput. 

\subsubsection*{Unhappy Path}\label{sec:zaptos:discussion:unhappy}
We discuss the scenario in \name when a block proposal is {\em not} ordered by consensus. 
There are two implications.
\begin{itemize}
    \item In cases where a block is optimistically executed but not eventually ordered by consensus, the computed block state becomes obsolete, wasting computation and potentially degrading system performance. 
    However, in such cases, during consensus, the proposer (leader) of the block times out, causing a more substantial performance degradation in consensus due to the lack of progress during the timeout. 
    Thus, the waste of computational resources is not a primary performance bottleneck in such scenario. 
    \item In cases where a block is optimistically committed but not eventually ordered by consensus, different honest validators or fullnodes may end up with inconsistent storage. This can occur if the leader is malicious or slow, or if the network becomes asynchronous, resulting in only a subset of honest servers optimistically executing and committing the block, which is eventually orphaned in consensus.
    To ensure data consistency across servers, any optimistically committed block will be reverted from storage if it is not eventually ordered by consensus.
\end{itemize}

\subsubsection*{Block Execution Time Calibration}\label{sec:zaptos:discussion:calibration}
We discuss how the blockchain system calibrates the execution time of each block under high load.
This is achieved through two techniques: {\em instruction-level gas calibration}~\cite{eth-gas,apt-gas} for smart contract transactions, and {\em block-level gas limit}~\cite{eth-blockgaslimit,apt-blockgaslimit} for blocks.
The gas assigned to each smart contract instruction measures its execution cost, which includes the cost of computation, I/O and storage. 
The total gas cost of a smart contract transaction is the sum of the gas costs of all its instructions. 
The block gas limit sets system parameters that cap the maximum gas allowed for computation, I/O and storage with a single block. 
During block execution, transactions are finalized~\footnote{In parallel execution~\cite{gelashvili2023block}, the execution of transactions can happen in parallel but the finalization of transactions' executions are sequential.} in the order they appear in the block.
When the accumulated gas of finalized transactions exceeds the block gas limit, block execution stops, and the executed transactions up to that point updates the new blockchain state.
To calibrate block execution time, the gas costs for individual instructions must first be accurately calibrated to reflect actual computation time. Then, the block gas limit is set accordingly, ensuring that block execution will halt at the desired time duration once the gas limit is met.

\subsubsection*{Extension to Client-Validator Model}\label{sec:zaptos:discussion:extension}
Some blockchain systems, such as Solana~\cite{solana}, allow clients to communicate directly with validators. In this model, \name's optimizations (opt-execution, opt-commit, and certification optimization) can still be applied to the validators' pipeline, achieving a latency reduction of 3 rounds when the execution and commit stage latencies exceed one round (similar to~\Cref{cor:latency}).

\section{Implementation}
\label{sec:impl}

We implement the pipelined architectures of \name (\Cref{alg:zaptos})~\footnote{\url{https://github.com/aptos-labs/aptos-core/tree/daniel-paper}} 
in {\sf Rust}~\cite{rust}, atop the open-sourced Aptos blockchain codebase~\cite{aptos}.
We collaborated with the Aptos Labs team to integrate the event-driven pipeline design (\Cref{alg:baseline}) into the Aptos blockchain, and are currently working with them to deploy \name (\Cref{alg:zaptos})~\cite{aptos-pipeline}.

The implementation uses {\sf Tokio}~\cite{tokio} for asynchronous networking, {\sf blstrs}~\cite{blstrs} for cryptography.
The primitives for consensus, execution and storage are mentioned in~\Cref{sec:prim}. 
Next, we highlight several key aspects for modular and efficient implementation.

\subsubsection*{Modularity and Extensibility}
The implementation generally follows the pseudo-code, modularizing the pipeline into distinct stages with well-defined interfaces between directly dependent stages, as illustrated in~\Cref{fig:zaptos} and~\Cref{fig:pipeline-zaptos}. 
These dependencies are managed using {\sf Rust}'s {\sf Future} feature~\cite{rust-future}, enabling asynchronous computation to produce the output of each stage so that subsequent stages can start as soon as the future is fulfilled.
Leveraging {\sf Rust}’s asynchronous programming, the implementation supports concurrent execution of pipeline stages, maximizing resource utilization for high throughput.
This modular design not only promotes high performance but also enables extensibility: new stages can be added by adjusting the interfaces of adjacent stages without modifying the contents of other stages or the rest of the pipeline.
Finally, the pipeline implementation can be reused by both validators and fullnodes, with each selecting specific stages to include.

\subsubsection*{More Stages}
In the actual implementation, the pipeline stages are further subdivided to improve parallelism and reduce latency. For example, the execution stage is divided into two sub-stages: executing the block to generate the new blockchain state, and merklizing the new state (computing the Merkle Tree proofs). 
This subdivision allows the execution stage of a block to begin as soon as the first sub-stage of the parent block's execution stage completes, thereby reducing pipeline latency.
The implementation also includes additional stages not detailed in~\Cref{alg:zaptos}. For instance, prior to the execution stage, an execution preparation stage retrieves the transaction payloads corresponding to metadata ordered by the consensus stage; following the commit stage, a post-commit stage is responsible for notifying other system components about committed transactions to prevent duplication.

\subsubsection*{Error Handling}
In the pipeline, each stage may fail with recoverable errors or unrecoverable errors.
If a stage fails with a recoverable error, the stage itself will be retried. 
When a stage fails with an unrecoverable error, subsequent stages skip execution and only propagate the error. 
The validator will handle the error depending on its type. For example, for an unrecoverable error, the validator may fallback to syncing the latest committed state from other validators.

\section{Evaluation}
\label{sec:eval}

In this section, we evaluate the throughput-latency performance of \name with geo-distributed experiments, and compare it with the \baseline (\Cref{alg:baseline}).
For fair comparison, the implementation and evaluation use the same primitives for each pipeline stage, as described in~\Cref{sec:prim}.

\subsection{Setup}\label{sec:eval:setup}

\paragraph{Machine Specifications.}
We used virtual machines from the Google Cloud Platform for the experiments. We used two variants of \texttt{n2d} series machines, namely \texttt{n2d-standard-32} with 32 cores and 128 GB memory and \texttt{n2d-standard-64} with 64 cores and 256 GB memory. This is to show that the execution is compute-bound and the overall system throughput is bounded by execution -  increasing compute increases throughput. Each machine has a 2TB network attached disk to guarantee enough IOPS for persistence.

\paragraph{Geo-distribution.}
The testbed consists of virtual machines in 10 different regions to mimic a globally decentralized network. This includes 2 regions in the US
(us-west1 and us-east1), 2 regions in Europe (europe-west4
and europe-southwest1), 1 region each in South America
(southamerica-east1), South Africa (africa-south1) and Australia (australia-southeast1), and 3 in Asia (asia-northeast3,
asia-southeast1, and asia-south1). The round-trip times between these regions range between 25ms and 317ms. For our experiments, we deployed 100 validators and 30 fullnodes evenly across 10 regions.

In our evaluation, we simplify the experiments by assuming constant and negligible latencies for both client-to-fullnode and fullnode-to-validator communication, as these latencies are identical for \name and \baseline and do not impact the comparative results.
Therefore, we let each fullnode and client to be co-located with a validator.

\paragraph{Metrics.}
We measure the blockchain's end-to-end latency as the fullnode end-to-end latency, measured from the fullnode receives the client's transaction to fullnodes commits the client's transaction, as the client-to-fullnode latency is negligible. Unless otherwise stated, we measure the 50th percentile latencies.

\paragraph{Workload.}
Each client transaction consists of 300-byte payload encoding peer-to-peer transfer from a source account to a target account along with metadata such as max gas fees and expiration timestamp. The clients send transactions in an open-loop at the target throughput rate.

\begin{figure}
    \begin{subfigure}{0.5\textwidth}
        \centering
        \includegraphics[width=\columnwidth]{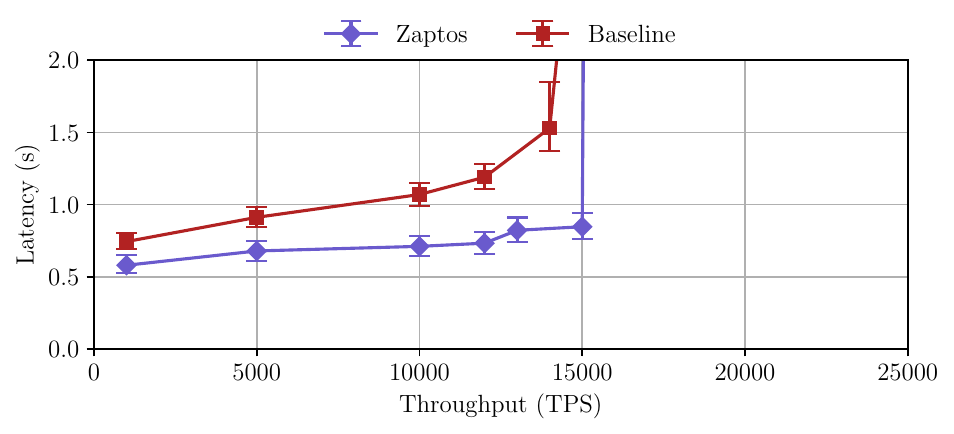}
        \caption{32 CPUs}
        \label{fig:commoncase-e2e-perf-32}
    \end{subfigure}
    \begin{subfigure}{0.5\textwidth}
        \centering
        \includegraphics[width=\columnwidth]{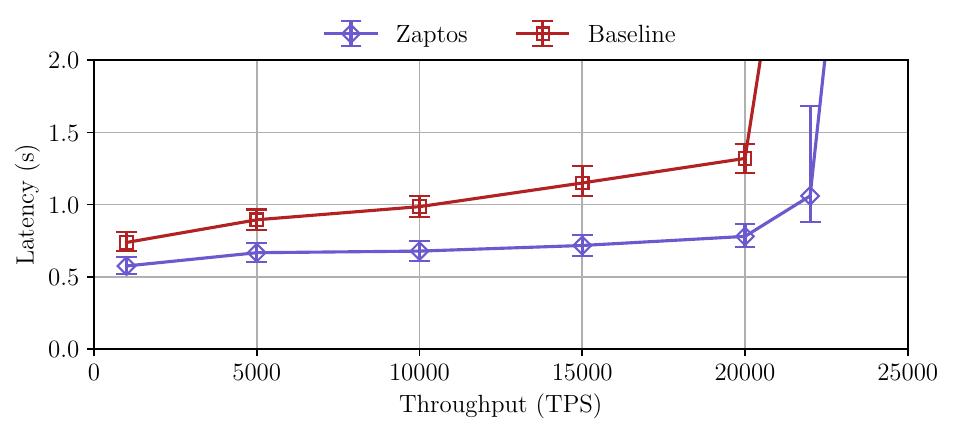}
        \caption{64 CPUs}
        \label{fig:commoncase-e2e-perf-64}
    \end{subfigure}
    \caption{Common case performance of \name and \baseline.}
    \label{fig:commoncase-e2e-perf}
\end{figure}

\subsection{Common Case}\label{sec:eval:common}
We first compare the performance of \name and \baseline in the common case, where the validators, fullnodes and clients are non-faulty, and the network is well connected. 
We use the terms latency and \latency interchangeably to refer to the blockchain's end-to-end latency.

\subsubsection*{Throughput-latency Graph}\label{sec:eval:common:Lgraph}

\Cref{fig:commoncase-e2e-perf} shows the latency with respect to throughput graph of \name and \baseline in the common case. The vertical bars for each data point represent the 25th and the 75th percentile latencies.
As depicted, both systems have a gradual increase in \latency as the system load grows. 
However, when reaching maximum capacity, \latency spikes sharply due to a pronounced rise in block queuing latency.
For instance, on 32-cpu machines, \baseline's latency rises from 0.75s at 1k TPS to 1.53s at 14k TPS, while \name's latency increases from 0.58s at 1k TPS to 0.85s at 15k TPS.
Similarly, on 64-cpu machines, \baseline's latency grows from 0.74s at 1k TPS to 1.32s at 20k TPS, whereas \name's latency increases from 0.58s at 1k TPS to 0.78s at 20k TPS. 
Across both machine types, \name outperforms \baseline as expected, significantly reducing \latency by 160 ms under low load and over 0.5 second under high load.
Notably, \name achieves subsecond \latency at 20k TPS with a production-grade implementation tested in a realistic, mainnet-like environment.

\subsubsection*{Latency Breakdown}\label{sec:eval:common:breakdown}

\begin{figure*}
    \begin{subfigure}{0.5\textwidth}
        \centering
        \includegraphics[width=\columnwidth]{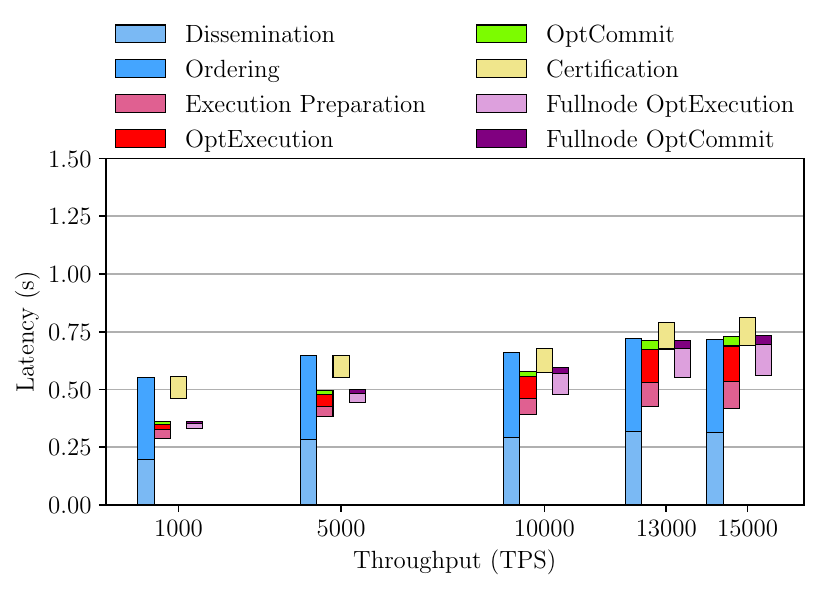}
        \caption{\name}
        \label{fig:breakdown-zaptos-32}
    \end{subfigure}%
    \begin{subfigure}{0.5\textwidth}
        \centering
        \includegraphics[width=\columnwidth]{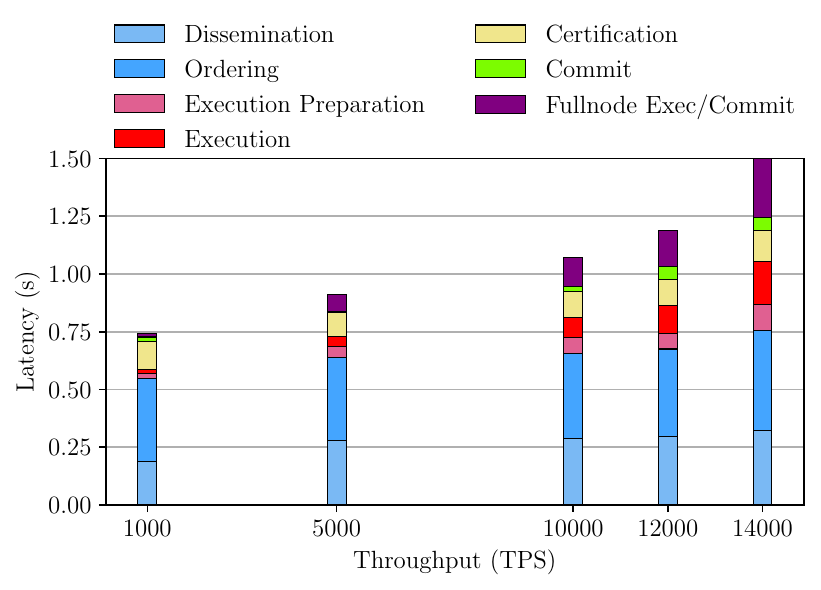}
        \caption{\baseline}
        \label{fig:breakdown-baseline-32}
    \end{subfigure}
    \caption{Latency breakdown of common case: 32 CPUs.}
    \label{fig:breakdown-32}
\end{figure*}

\begin{figure*}
    \begin{subfigure}{0.5\textwidth}
        \centering
        \includegraphics[width=\columnwidth]{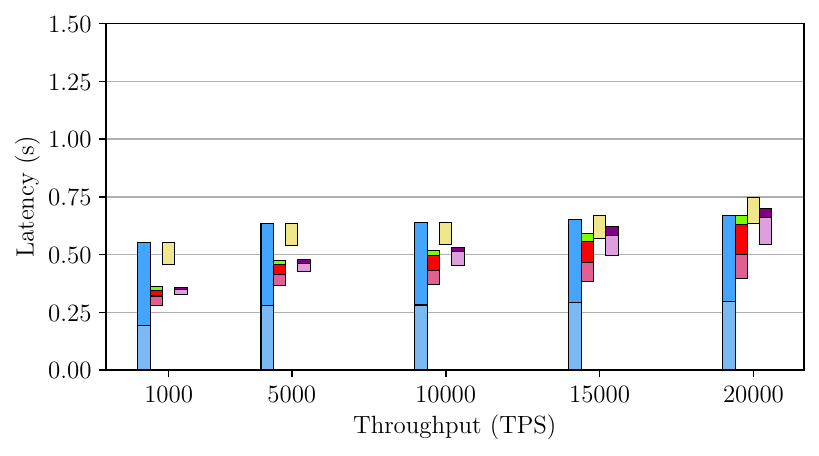}
        \caption{\name}
        \label{fig:breakdown-zaptos-64}
    \end{subfigure}%
    \begin{subfigure}{0.5\textwidth}
        \centering
        \includegraphics[width=\columnwidth]{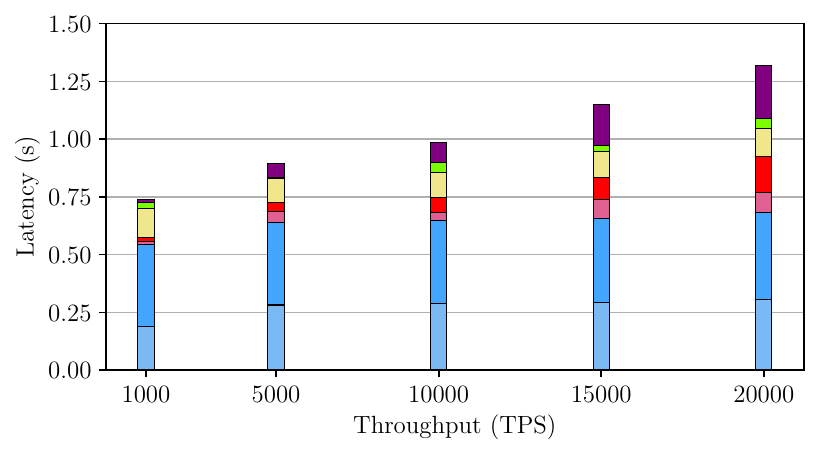}
        \caption{\baseline}
        \label{fig:breakdown-baseline-64}
    \end{subfigure}
    \caption{Latency breakdown of common case: 64 CPUs.}
    \label{fig:breakdown-64}
\end{figure*}

To further analyze system performance and validate the optimizations in \name, in Figures~\ref{fig:breakdown-32} and~\ref{fig:breakdown-64}, we provide a detailed latency breakdown of the data points from the throughput-latency graph (\Cref{fig:commoncase-e2e-perf}). 
The latency breakdown graph depicts the duration of each pipeline stage for both validators and fullnodes. 

For \baseline, the latency breakdown includes the following stages: consensus dissemination, consensus ordering, execution preparation, execution, certification, commit, and fullnode stages including execution preparation, execution and commit (grouped and represented as a single bar in the graph for simplicity). 
The execution preparation stage involves fetching the transaction payloads corresponding to metadata ordered by the consensus stage, and waiting for the execution stages of ancestor blocks to finish. 
In \baseline, all stages are processed sequentially for each block, as illustrated in~\Cref{fig:breakdown-baseline-64}.

In contrast, \name introduces parallel pipelining. 
Similarly, its latency breakdown includes the following stages: consensus dissemination, consensus ordering, execution preparation, optimistic execution, certification, optimistic commit, and fullnode optimistic execution and optimistic commit. 
However, in \name, stages for the same block can overlap and execute in parallel as long as their preconditions are met, as detailed in~\Cref{alg:zaptos}. Dependencies between stages are illustrated in~\Cref{fig:zaptos}. 
For instance, receipt of a block proposal triggers execution preparation and subsequently opt-execution;
and opt-execution completion triggers opt-commit, and certification once an \tags{OrderVote} is sent (last round of consensus).

From Figures~\ref{fig:breakdown-32} and~\ref{fig:breakdown-64}, we make the following key observations:
\begin{itemize}
    \item The \latency of \name is approximately equal to the consensus latency, up to 5k TPS on 32-cpu machines and 10k TPS on 64-cpu machines. During this range, the execution, certification and commit stages for both validators and fullnodes are effectively "shadowed" within the consensus stage. This partially validates~\Cref{thm:latency}.
    
    For instance, \baseline on 64-cpu machines at 10k TPS, the non-consensus stages contribute 338 ms, which is more than 50\% of the consensus latency of 648ms (dissemination latency: 288ms, ordering latency: 360ms). In contrast, \name achieves a 30\% reduction in \latency by compressing the pipeline. 
    
    \item As TPS increases, the non-consensus stages are no longer fully shadowed within the consensus stage. 
    This is primarily due to the increased execution preparation required to wait or fetch larger blocks, and the longer duration of the opt-execution stage. 
    When execution becomes the primary bottleneck for TPS under heavy load, increasing the number of CPUs from 32 to 64 enables both \baseline and \name to achieve higher peak TPS.
    While the certification stage still completes within a single round, it can only begin after the opt-execution is finished. 
    The opt-commit stage remains relatively short, since the peer-to-peer transaction workload used in evaluation is not commit-intensive. 
    However, for workloads involving transactions that write substantial amounts of on-chain data, the opt-commit stage could emerge as a bottleneck as TPS increases.
    For fullnodes, the opt-execution and opt-commit stages follow shortly after those of the validators, effectively reducing the overall pipeline duration. 
    
    Despite partial overlap of the stages at maximum throughput, \name significantly reduces latency by shadowing majority of the stage durations.
    For instance, at 20k TPS on 64-core machines, 
    \baseline exhibits a total latency of 1.32s (consensus latency: 0.68s; other stages: 0.64s), whereas \name reduces this to 0.78s (consensus latency: 0.67s; other stages: 0.11s).
    
    \item The consensus dissemination latency remains a bottleneck of the \latency under high load. As described in~\Cref{sec:prim:consensus}, this latency comprises the following:
    Batch queuing latency which is the expected time for a transaction to be included in a payload batch (\textasciitilde 110ms). This is caused by an intentionally high payload batch creation interval (\textasciitilde 200ms) to prevent network overload;
    Block queuing latency which is the expected time for a transaction payload batch to be included in a block (\textasciitilde 80ms), arising from a block interval of \textasciitilde 150ms;
    Broadcast latency which is the time required for one round of broadcasting the payload batch (\textasciitilde 120ms).
    These components collectively result in dissemination latencies of \textasciitilde 300ms under load. 
    Further improving the dissemination latency presents a interesting engineering challenge for future work.
    
\end{itemize}

\subsection{Under Byzantine Failures}\label{sec:eval:failure}

\begin{figure}
    \centering
    \includegraphics[width=\columnwidth]{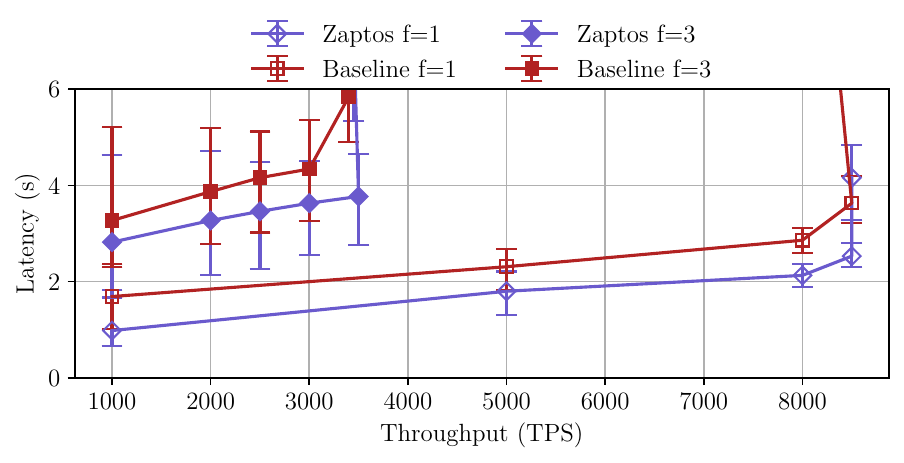}
    \caption{Failure case performance of \name and \baseline, under leader equivocation and round-robin leader rotation (leader reputation~\cite{cohen2022aware} disabled).}
    \label{fig:failures}
\end{figure}

We demonstrate that \name achieves better average system performance even in presence of failures. 
The evaluation focuses on the leader equivocation failure scenario, where a faulty leader sends distinct block proposals to different validators. 
As discussed in~\Cref{sec:zaptos:discussion:unhappy}, leader equivocation may result in resource wastage when a block is opt-executed and opt-committed but eventually not ordered. 
A partitioned or slow leader can be regarded as a special case of an equivocating leader, where the faulty leader sends its proposal to only a subset of the validators rather than the entire network.
In the case of crash failures, \name improves average performance, performing comparably to \baseline under crashed leaders and outperforming it under non-faulty leaders.
Consequently, we focus our evaluation on scenarios with equivocating leaders.

We compare \name and \baseline in a small-scale network comprising of 10 validators and 10 fullnodes distributed evenly across 10 regions, as described in~\Cref{sec:eval:setup}. 
Our evaluation includes scenarios with 1 and 3 faulty validators acting maliciously by proposing distinct blocks to different validators during their turn as leaders. For these tests, we use \texttt{n2d-standard-32} machines (32 cores, 128 GB memory).
To reflect the impact of malicious leaders, we use a naive round-robin leader rotation schedule. 

As shown in~\Cref{fig:failures}, under equivocating leaders and round-robin leader rotation, the system performances of both \name and \baseline degrade significantly, compared to the common case (\Cref{sec:eval:common}). 
Specifically, the peak throughput of both systems decreases to approximately 8k TPS with 1 faulty validator and 3.5k TPS with 3 faulty validators. 
In this case, \name also demonstrates significantly better average \latency than \baseline. 
For example, with one faulty validator, \name achieves a latency improvement of approximately 0.5s to 0.7s compared to \baseline. Similarly, with three faulty validators, the latency improvements range from around 0.45s to 0.7s.
This improvement is attributed to the fact that during rounds with faulty leaders, both systems experience consensus stalling as the primary bottleneck. However, in rounds with non-faulty leaders, \name achieves a substantial reduction in pipeline latency.

While leader reputation protocols~\cite{cohen2022aware} are commonly employed in real-world deployments to select high-performing leaders and enhance system stability, we intentionally disable them in the evaluation above to focus on performance under failure scenarios.
Our additional testing with the leader reputation protocol shows that once the mechanism is activated to filter out faulty leaders, the system's performance recovers to levels comparable to the common case. Therefore, we do not report the numbers in the paper.

\section{Related Work}
\label{sec:rw}

\subsubsection*{Blockchain architectures.}\label{sec:rw:blockchain}
We classify existing blockchain architectures into three categories based on how their pipeline stages interact: coupled-consensus-execution, execution-then-consensus, and consensus-then-execution. In the latter two categories, consensus and execution are decoupled as separate stages. 

In the coupled-consensus-execution architecture (illustrated in~\Cref{fig:coupled}), the consensus stage is tightly integrated with the execution of blocks which determines the new blockchain state during consensus. For instance, in leader-based protocols, validators execute a block after the leader's proposal and vote on the resulting new blockchain state. The output of the consensus stage includes the finalized state. 
Representative chains that use this architecture include Ethereum PoS~\cite{ethereumpos}, Solana~\cite{solana}, Algorand~\cite{gilad2017algorand}, Cosmos~\cite{cason2021design}, Redbelly~\cite{crain2021red}, NEAR~\cite{near} and Bitcoin~\cite{nakamoto2008bitcoin}. XRP~\cite{xrp} and Stellar~\cite{mazieres2015stellar} do not organize transactions as blocks but they also couple consensus together with execution of transactions.

The execution-then-consensus architecture is first introduced in HyperLedger~\cite{androulaki2018hyperledger}. As illustrated in~\Cref{fig:execution-then-consensus}, validators first execute a list of transactions locally, producing execution outputs. These outputs are then subjected to a consensus process to agree on their ordering and, consequently, the new blockchain state. 

In the consensus-then-execution architecture (illustrated in~\Cref{fig:consensus-then-execution}), validators initially reach a consensus on a new block extending the blockchain. Execution of the ordered block follows, producing the updated blockchain state. Ideally, deterministic execution ensures consistency across validators. However, non-deterministic execution caused by software bugs or hardware issues can lead to inconsistent states. 
To address this issue, some systems introduce a certification stage prior to commit, as seen in Diem~\cite{diem-whitepaper}, Aptos~\cite{aptos-whitepaper} and Sui~\cite{blackshear2024sui}. This stage involves validators collectively signing the new state. Alternatively, other systems, such as Porygon~\cite{chen2024porygon} (illustrated in~\Cref{fig:porygon}), run a second consensus on the new state by piggybacking on subsequent consensus instances. 
The first approach has lower latency, but compromises liveness when non-deterministic execution occurs, and no majority of the validators sign the same state. 
Both approaches produce a publicly verifiable proof of the new blockchain state, enabling clients to confirm that their transaction has been successfully committed to the blockchain.
In Sui Lutris~\cite{blackshear2024sui}, an additional consensusless path is introduced to bypass the consensus stage for a subset of transactions involving only single-writer operations. 

To the best of our knowledge, the Diem blockchain~\cite{diem-decouple} was the first to design and implement a pipelined architecture, which has since been adopted by several major blockchains, including Aptos~\cite{aptos}, Avalanche~\cite{ava-decouple} and Sui~\cite{sui}.
This paper employs \baseline as the baseline for \name.
Recently, Porygon~\cite{chen2024porygon} proposes a 3D parallelism architecture aimed at improving blockchain system throughput and latency, which includes a transaction processing pipeline that pipelines the consensus, execution, and commit stages. 
The key differences between the pipeline architecture of Porygon~\cite{chen2024porygon} and \baseline~\cite{diem-decouple} are as followings:
(1) The \baseline~\cite{diem-decouple} uses a pipelined consensus protocol (an improved version of Jolteon~\cite{gelashvili2022jolteon, optqs, optjolteon}), which pipelines internal consensus stages to reduce the number of consensus messages and block interval (2 rounds). 
    In contrast, Porygon~\cite{chen2024porygon} employs a non-pipelined consensus protocol (BA*~\cite{gilad2017algorand}) and has a block interval at least equal to the full consensus ordering latency (4 rounds). It is because without pipelining, each consensus instance begins only after the previous instance finishes. 
    The longer block interval leads to higher blockchain end-to-end latency, as the expected queuing delay for consensus increases with block interval duration.
(2) In Porygon~\cite{chen2024porygon}, during the certification stage (part of the commit stage in their paper), the updated blockchain state after execution is piggybacked with the next consensus block proposal, for validators to reach agreement. 
    This approach results in the certification stage taking as long as the consensus stage. In contrast, the \baseline handles certification in a single round, significantly reducing latency.

\subsubsection*{Consensus latency.}\label{sec:rw:consensus}
As mentioned in~\Cref{sec:prim:consensus}, consensus latency consists of consensus dissemination latency and consensus ordering latency. 
A significant body of research has focused on reducing the consensus ordering latency of Byzantine fault tolerant (BFT) consensus protocols under partial synchrony, which can broadly be categorized into leader-based and DAG-based approaches.

For leader-based protocols, PBFT~\cite{castro1999practical} pioneered the field with an optimal consensus ordering latency of 3 rounds, though it lacked efficient data dissemination, leader rotation, and pipelining.
Decoupling data dissemination is crucial for maximizing consensus throughput, as it allows reaching agreement solely on the metadata. 
In the blockchain era, leader rotation is cruicial for ensuring fairness and aligning economic incentives within the network. 
Additionally, pipelining reduces message complexity and the variety of message types, thereby simplifying the protocol's implementation.
Since PBFT, various improvements have been made to address these limitations~\cite{buchman2016tendermint,buterin2017casper,yin2019hotstuff,gelashvili2022jolteon,doidge2024moonshot,danezis2022narwhal}. HotStuff~\cite{yin2019hotstuff}, in particular, achieves pipelining and leader rotation with linear message complexity, at the expense of increased consensus ordering latency. 
Subsequent works~\cite{gelashvili2022jolteon, doidge2024moonshot, optjolteon} improve the consensus ordering latency.
DAG-based protocols are inherently designed for efficient data dissemination, seamless leader rotation, and effective pipelining. Recent works~\cite{keidar2021all,danezis2022narwhal,spiegelman2022bullshark,spiegelman2023shoal,babel2023mysticeti,arun2024shoal++} have focused on reducing latency while achieving high throughput.

It is well known that the lower bound of the consensus ordering latency is 3 rounds under partial synchrony and optimal fault tolerance~\cite{kuznetsov2021revisiting,abraham2021good}, and many leader-based BFT protocols~\cite{castro1999practical,doidge2024moonshot} achieves this optimal 3-round consensus ordering latency. 
However, these protocols require an additional round of dissemination to include transactions (or metadata of transaciton batches) from all validators in the proposal to achieve high throughput, resulting in a total consensus latency of 4 rounds.
Similarly, in the latest DAG-based BFT protocols~\cite{keidar2022cordial,babel2023mysticeti}, the latency includes an implicit 1 round for dissemination followed by 3 rounds for ordering. 

A multi-leader optimization has been proposed in several studies for both leader-based protocols~\cite{stathakopoulou2019mir} and DAG-based protocols~\cite{spiegelman2023shoal}, reducing dissemination latency by ordering proposals from multiple validators for the same round. Ideally, this approach could eliminate dissemination latency entirely by treating all validators as leaders. However, in practice, a large multi-leader setup is not resilient: if any validator in the set is slow or malicious, the entire set of proposals may fail. Consequently, even a small subset of slow or malicious validators could severely degrade system performance or even compromise liveness.

\section{Conclusion}

This paper introduces \name, a novel blockchain pipeline architecture designed to achieve low latency by preemptively parallelizing pipeline stages, and high throughput by maximizing resource utilization through effective pipelining.

\begin{acks}
We thank the Aptos Labs Research team and Engineering team for valuable discussions and feedbacks, especially Rati Gelashvili, Alden Hu, Igor Kabiljo and Joshua Lind.
\end{acks}

\bibliographystyle{ACM-Reference-Format}
\bibliography{ref}

\clearpage
\appendix


\section{Existing Blockchain Pipeline Architectures}
\label{app:pipeline}

We provide illustrations of existing blockchain pipeline architecture, including coupled-consensus-execution in~\Cref{fig:coupled}, execution-then-consensus in~\Cref{fig:execution-then-consensus}, and consensus-then-execution in~\Cref{fig:consensus-then-execution} and~\Cref{fig:porygon}.
More details of can be found in~\Cref{sec:rw:blockchain}.

\begin{figure}[t!]
    \centering
    \includegraphics[width=0.8\linewidth]{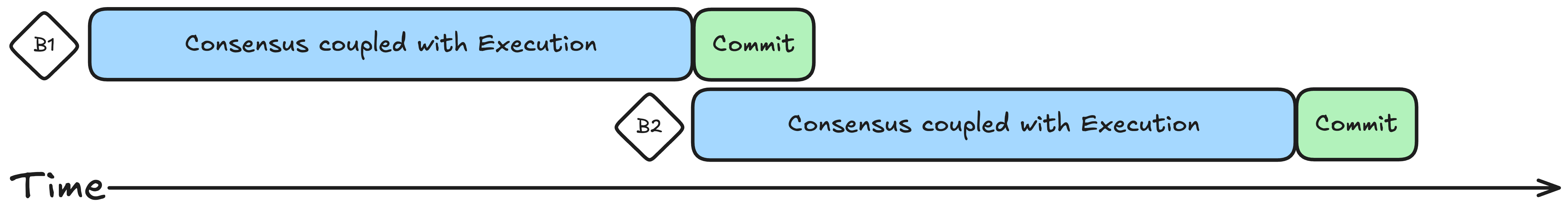}
    \caption{Illustration of the coupled-consensus-execution pipeline architecture. The consensus stage is tightly integrated with the execution of blocks which determines the new blockchain state during consensus.}
    \label{fig:coupled}
\end{figure}

\begin{figure}[t!]
    \centering
    \includegraphics[width=\linewidth]{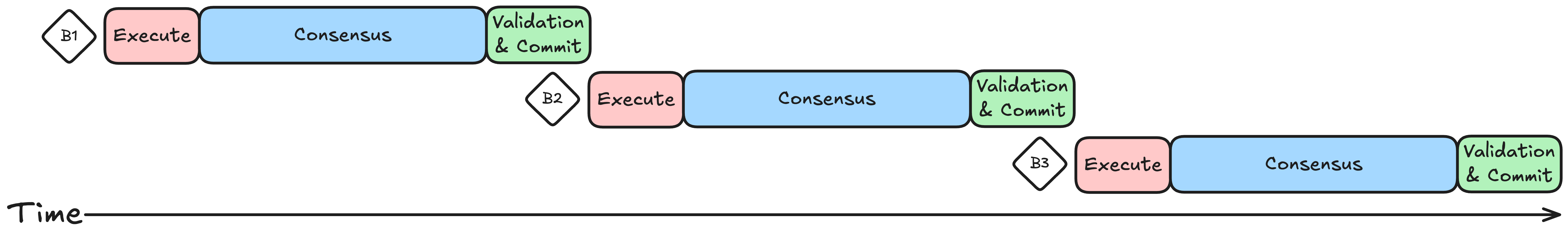}
    \caption{Illustration of the execution-then-consensus pipeline architecture. Validators first execute a list of transactions locally, producing execution outputs. These outputs are then subjected to a consensus process to agree on their ordering and, consequently, the new blockchain state.}
    \label{fig:execution-then-consensus}
\end{figure}

\begin{figure}[t!]
    \centering
    \includegraphics[width=0.8\linewidth]{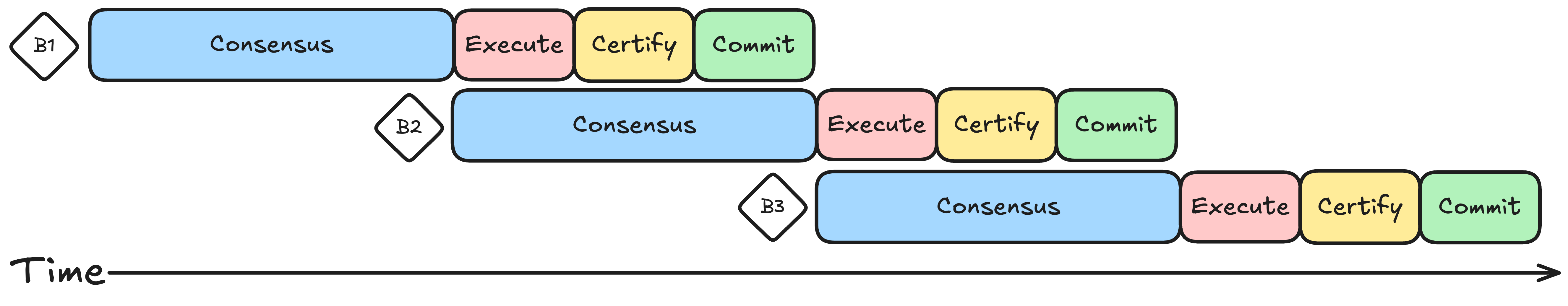}
    \caption{Illustration of the consensus-then-execution pipeline architecture. 
    Validators initially reach a consensus on a new block extending the current blockchain. Execution of the ordered block follows, producing the updated blockchain state. To produce publicly verifiable proof of the new blockchain state and avoid safety violation caused by non-deterministic execution, a certification stage is introduced prior to commit. When integrated with a pipelined consensus protocol, this architecture effectively becomes the pipelined architecture of \baseline (\Cref{fig:baseline}). }
    \label{fig:consensus-then-execution}
\end{figure}

\begin{figure}[t!]
    \centering
    \includegraphics[width=\linewidth]{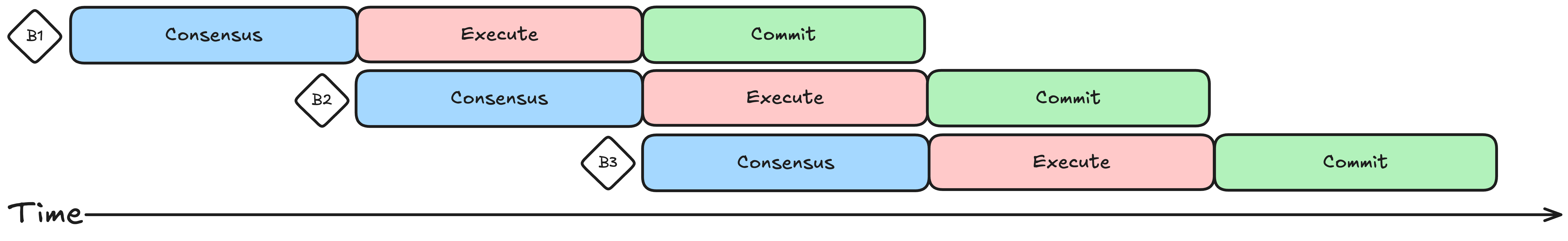}
    \caption{Illustration of the Porygon~\cite{chen2024porygon} pipeline architecture. Porygon falls under the consensus-then-execution pipeline architecture, but relies on subsequent consensus stage for certification. 
    As a result, a block $\blk_1$ is committed only after the consensus instance for $\blk_3$, which includes the piggybacked updated blockchain state following its execution, is completed.}
    \label{fig:porygon}
\end{figure}

\end{document}